\documentclass[final,5p,times]{elsarticle}

\usepackage{soul,framed}  
\usepackage[monochrome]{xcolor}

\usepackage{algorithm}
\usepackage{algpseudocode}

\colorlet{shadecolor}{yellow} 
\graphicspath{{../pdf/}{../jpeg/}}
\DeclareGraphicsExtensions{.pdf,.jpeg,.png}
\usepackage{amsfonts}
\usepackage{amssymb}
\usepackage{amsthm}

\usepackage{amsmath}
\usepackage{pifont}
\usepackage{balance}
\usepackage{longtable}

\newtheorem{thm}{Theorem}
\newtheorem{lem}[thm]{Lemma}

\theoremstyle{definition}

\usepackage{array}
\usepackage{mdwmath}
\usepackage{mdwtab}
\usepackage{eqparbox}
\usepackage{url}
\usepackage{booktabs}
\usepackage{tabularx}  
\usepackage{subcaption}  

\usepackage{flushend}

\newdimen\slantmathcorr
\def\oversl#1{%assuming that mathslant=0.25
\setbox0=\hbox{$#1$}
\slantmathcorr=\wd0
\hskip 0.2\slantmathcorr \overline{\hbox to 0.8\wd0{%
\vphantom{\hbox{$#1$}}}}
\hskip-\wd0\hbox{$#1$}
}

\def\undersl#1{%assuming that mathslant=0.25
\setbox0=\hbox{$#1$}
\slantmathcorr=\wd0
\underline{\hbox to 0.8\wd0{
\vphantom{\hbox{$#1$}}}}
\hskip-0.8\wd0\hbox{$#1$}
}

\newcommand{\threshold}{\mathrm{threshold}}
\newcommand*\chem[1]{\ensuremath{\mathrm{#1}}}

\def\tsc#1{\csdef{#1}{\textsc{\lowercase{#1}}\xspace}}
\tsc{WGM}
\tsc{QE}

\begin{document}
\let\WriteBookmarks\relax
\def\floatpagepagefraction{1}
\def\textpagefraction{.001}

\title {BattBee: Equivalent Circuit Modeling and Early Detection of Thermal Runaway Triggered by Internal Short Circuits for Lithium-Ion Batteries}  

\author[label1]{Sangwon~Kang} \author[label1]{Hao~Tu} \author[label2]{Huazhen~Fang\corref{cor1}}
\address[label1]{Department of Mechanical Engineering, University of Kansas}
\address[label2]{Department of Mechanical Engineering, Michigan State University}
 \cortext[cor1]{Corresponding author. E-mail address: hfang@msu.edu.}

\begin{abstract}
Lithium-ion batteries are the enabling power source for transportation electrification. However, in real-world applications, they remain vulnerable to internal short circuits (ISCs) and the consequential risk of thermal runaway (TR). Toward addressing the challenge of ISCs and TR, we undertake a systematic study that extends from dynamic modeling to fault detection in this paper. First, we develop {\em BattBee}, the first equivalent circuit model to specifically describe the onset of ISCs and the evolution of subsequently induced TR. Drawing upon electrochemical modeling, the model can simulate ISCs at different severity levels and predict their impact on the initiation and progression of TR events. With the physics-inspired design, this model offers strong physical interpretability and predictive accuracy, while maintaining structural simplicity to allow fast computation. Then, building upon the BattBee model, we develop fault detection observers and derive detection criteria together with decision-making logics to identify the occurrence and emergence of ISC and TR events. This detection approach is principled in design and fast in computation, lending itself to practical applications. Validation based on simulations and experimental data demonstrates the effectiveness of both the BattBee model and the ISC/TR detection approach. The research outcomes underscore this study's potential for real-world battery safety risk management. 
\end{abstract}

\begin{graphicalabstract}
\includegraphics[width=1\textwidth]{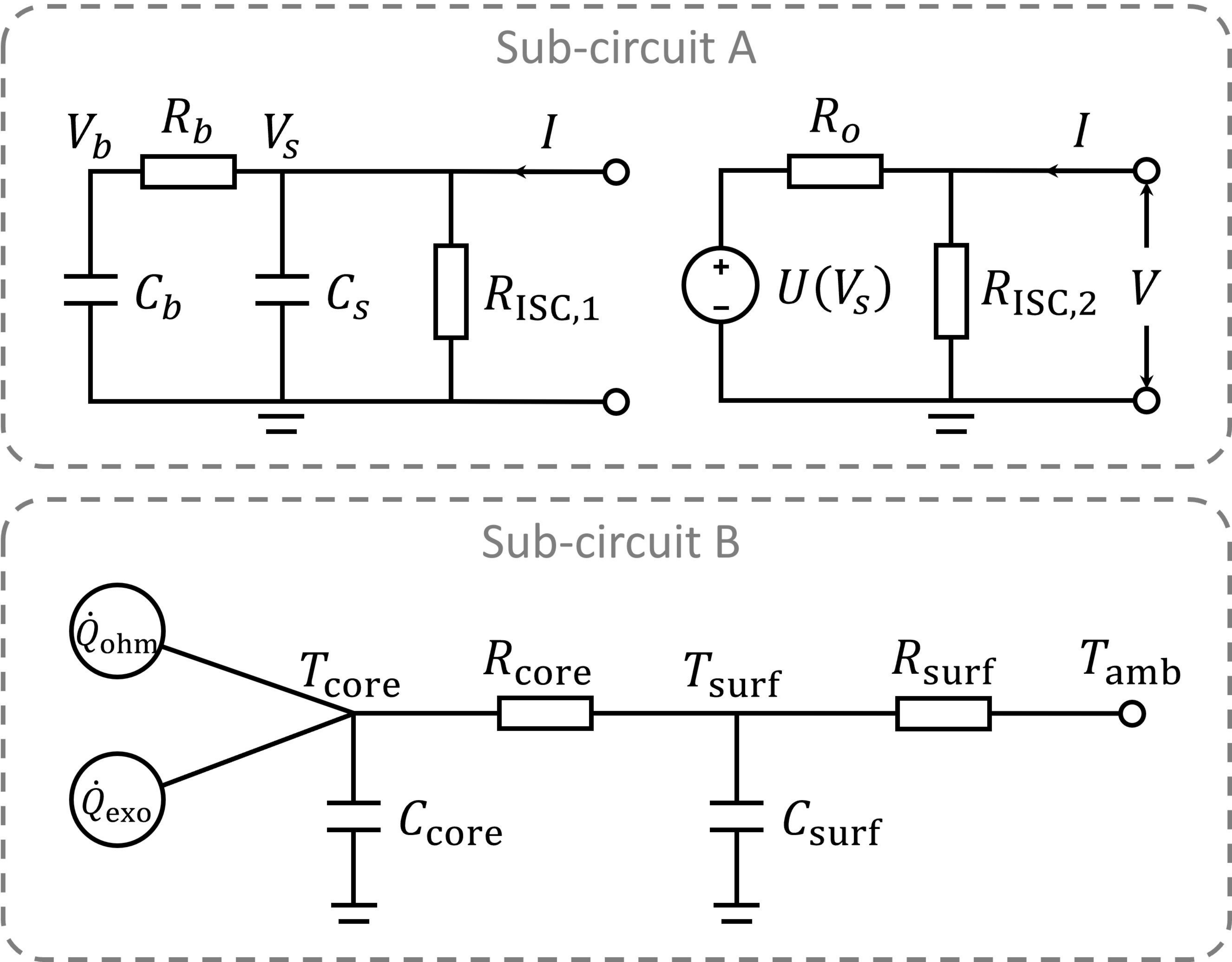}
\end{graphicalabstract}

\onecolumn
\begin{highlights}

\item BattBee models LiB's electro-thermal dynamics under ISC-triggered TR.
\item Model-based framework enables early ISC and TR detection with analytical thresholds.
\item Simulation and experimental results validate the model and the detection framework. 
\end{highlights}
\twocolumn

\maketitle

\section{Introduction} \label{sec:I. Introduction}
Lithium-ion batteries (LiBs) have been propelling forward the global clean energy transition. They power millions of electric vehicles, enable grid-scale energy storage, and facilitate renewable integration into power systems. Despite their significant merits including high energy density and long cycle life, LiBs have yet to be fully safe. The foremost challenge is thermal runaway (TR), a cascading sequence of chemical reactions that causes an uncontrollable temperature rise resulting in fires. Although low in probability, TR-induced fires have occurred to various electric vehicle models and stationary energy storage installations~\cite{Sun:FT:2020,UT-Fire-2025,EPRI-Fire-2025}. These incidents have raised a pressing question: How to ensure the safety of LiBs in practical applications? To address this challenge, we designed a new equivalent circuit model, named BattBee\footnotemark, to capture the electro-thermal behavior in TR triggered by internal short circuits (ISCs). Further, harnessing the model, we develop an observer-based TR detection approach. Combined, the BattBee model and fault detection approach offer an accurate and computationally fast method to deal with TR, thereby mitigating its impact on human safety and asset integrity.

\footnotetext{The model's name ``BattBee'' combines ``Batt'', short for battery, with ``Bee'', inspired by honeybees' reputed ability to detect explosives~\cite{Rodacy:SPIE:2002,Bromenshenk:JCWD:2003}.}
% https://storagewiki.epri.com/index.php/BESS_Failure_Incident_Database
% https://tools.utfireresearch.com/apps/incident_map

\begin{figure} [t]
 \centering
 \includegraphics[width=0.46\textwidth]{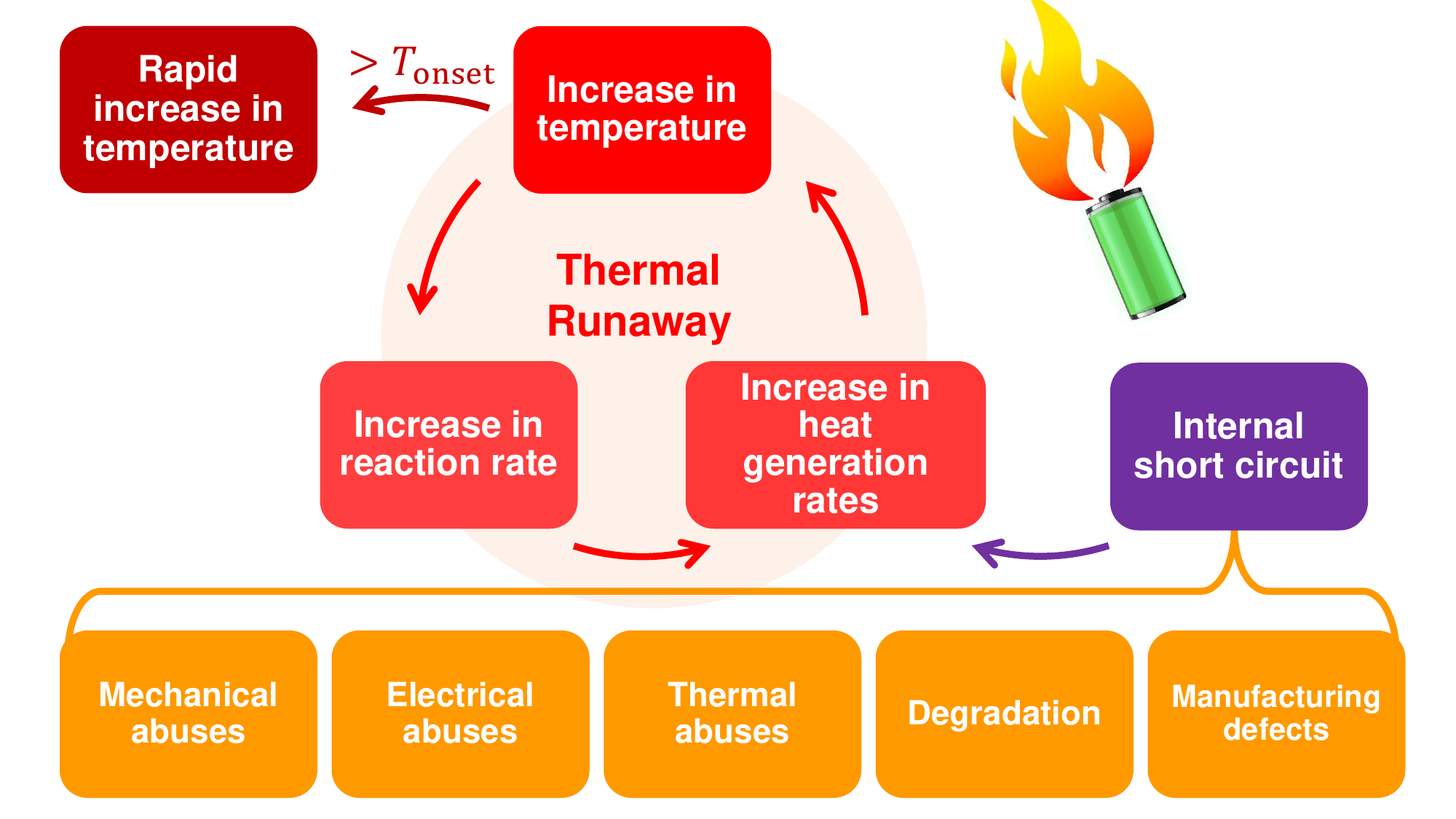}\\
\caption{Schematic for the formation of ISC-triggered TR ($T_{\mathrm{onset}}$ is the TR onset temperature).}
\label{fig:00}
\end{figure}

\subsection{Literature Review on TR and ISC Modeling}\label{Sec:Review-TR-ISC}

TR within a LiB cell springs from a complex process in which exothermic decomposition and accelerated heat build-up reinforce each other. The decomposition reactions break down the solid electrolyte interface (SEI) layers, the electrolyte, and cathode, releasing a significant amount of heat. Once the heating melts the separator, the ISC ensues to cause more heat generation. The rising temperature, in turn, increases the rates of the decomposition reactions, creating a vicious cycle. As a result, the cell's temperature escalates at an ever-increasing pace. During this process, the decomposition releases flammable gases such as ethene and methane from the electrolyte, liberates oxygen from the cathode, and builds up pressure within the cell, ultimately leading to fires and even explosions. TR can be caused by diverse abuse conditions, including: mechanical abuse (e.g., cell deformation or penetration due to collision and crush), electrical abuse (e.g., external short circuit, overcharge, and overdischarge), and thermal abuse (e.g., high ambient temperature, and local overheat)~\cite{Feng:ESM:2018,Spotnitz:JPS:2003}. 
Other major reasons include aging, degradation and cell defects in manufacturing. For almost all TR cases, ISCs represent a key trigger. They typically occur when the separator is ruptured or damaged to cause direct contact between the cell's cathode and anode---such separator damage can result from mechanical penetration or fracture, melting under high temperature, and piercing by lithium dendrites, due to abuse or degradation. On some occasions, ISCs trace to electrode misalignment in cell manufacturing or other reasons. In general, ISCs either directly initiate TR or drive an ongoing TR event toward an uncontrollable state if happening during this process~\cite{Feng:ESM:2018,Zhang:RSER:2021,Lin:JES:2023}.  Fig.~\ref{fig:00} shows the schematic of an ISC-triggered TR process.

Dynamic modeling is essential for understanding the occurrence and progression of TR, while providing a foundation for TR monitoring and detection. This subject has attracted a growing body of literature over the past decades. To our knowledge, the earliest study appears in~\cite{Hatchard:JES:2001}, which develops a 1D thermal model that accounts for the reversible heat generated by electrode reactions. However, this model only suits for predicting the onset of a TR event due to its limited consideration of exothermic decomposition. An extension of~\cite{Hatchard:JES:2001} to 3D thermal abuse modeling is shown in~\cite{Kim:JPS:2007}, which takes into account the cell shapes, the spatial distributions of materials and temperatures, and the impacts of electrolyte decomposition. Other 3D thermal abuse models appear in~\cite{Guo:JPS:2010,Lopez:JES:2015}, which consider different exothermic reactions across the SEI, cathode, anode, and electrolyte. The work in~\cite{Coman:JES:2017} investigates modeling for ISC-triggered TR, proposing an efficiency factor to capture the conversion of electric energy into thermal energy when an ISC happens. The literature also includes modeling for TR under mechanical abuse~\cite{Zhang:JPS:2015} and intertwined with venting of gas~\cite{Kim:JPS:2021,He:AE:2022}. These above studies focus on TR characterization by combining thermal modeling with changes in chemical reaction rates due to reaction kinetics. 

Further, coupling electrochemical dynamics with thermal failure allows to enable deeper understanding and modeling of TR. Taking a lead in this regard, the work in~\cite{Feng:JES:2018} comprehensively incorporates capacity decline under high temperature, ISC due to separator failure, and decomposition reactions that release significant heat for TR prediction. Various other electrochemical-thermal models have also emerged to describe TR caused by high ambient temperature, overcharge, and others~\cite{Ren:HPS:2017,Chakraborty:JES:2024,Mallarapu:JECS:2024}. If TR happens to a cell in a LiB pack, it will quickly spread to the adjacent cells in a domino effect, ravaging the pack and causing more hazards. To characterize TR propagation, existing studies have either coupled cell-level 0D or 2D models by accounting for heat conduction between cells~\cite{Feng:AE:2015,Bugryniec:ER:2020,Coman:AE:2017}, or developed 3D models that consider key mechanisms---such as mass loss and electrolyte vaporization---underlying TR propagation~\cite{Hoelle:JES:2023}.

Note that thermal modeling for LiBs under regular operating conditions is another topic of continuous interest. The flourishing literature encompasses a wide range of thermal models coupled with electrochemical, electrical, and mechanical behaviors, e.g.,~\cite{Gu-Wang:JES:2000,Tian:TCST:2019,Ai:JES:2020,Alkhedher:JES:2024,Lin:JPS:2014}, to name just several. Recent years have witnessed growing use of machine learning to achieve data-driven thermal modeling~\cite{Tu:CDC:2022,He-Wang:JES:2024,Yang-Nguyen:AE:2025}.

While many of the TR models reviewed above incorporate ISCs, a few other studies have specifically focused on ISC modeling for the purpose of detection or mechanistic understanding. A common approach involves introducing an ISC resistance into an equivalent circuit model, typically placed in parallel with the open-circuit voltage source~\cite{Ouyang:JPS:2015,Sun:JPS:2024,Lai:ElecActa:2018,Bhaskar:AE:2025}. This resistance takes a low value to capture the rapid charge depletion and voltage drop that occur during an ISC event. One can further combine such models with lumped thermal models to describe ISC-induced electro-thermal behaviors~\cite{Chen:ATE:2017,Yang:Auto:Innov:2022,Pan:JES:2025}. Beyond equivalent circuit modeling, some studies have rallied around pseudo 2D and 3D electrochemical-thermal modeling to characterize ISCs and their effects~\cite{Zavalis:JES:2012,Zhao:JES:2015,Cai:JES:2019,Kong:JES:2020,Liu:Energies:2022,Fang:JPS:2014}. These models provide insights into fundamental mechanisms governing the thermal and electrochemical responses during ISCs. 

\subsection{Literature Review on TR and ISC Detection}

Upon its initial formation, TR can quickly spiral into uncontrollable blazes to cause catastrophic failures. Early detection of this stealthy hazard is hence critical for real-world safety of LiB-powered applications and has become an urgent challenge. A straightforward way is to develop and deploy different types of sensors to monitor precursors indicating the emergence of TR~\cite{Koch:Batteries:2018,Kong:JES:2023,Zhang:Sensors:2023,Mei:NatComm:2023,Zhang:AE:2025}. Such precursors include abnormal changes in temperature, voltage, stress, and expansion force, as well as smoke and gas leakage. Useful as its is, this passive detection approach increases the cost, weight, and design complexity of LiB systems, with its accuracy contingent on sensor precision and susceptible to ambient operating conditions. By contrast, active detection seeks to integrate multiples sources of information, like different types of data and models, to identify TR in the nascent stage, promising earlier and more accurate detection. Existing methods generally fall into two categories: data-driven and model-based.

Data-driven detection methods attempt to extract features directly from measured data to find out anomalies leading to TR events. The work in~\cite{Li-Deng-Zhang-Liu-Wang:AE:2023} employs multi-dimensional statistical analysis to identify and select safety-representative features for LiBs, utilizing metrics such as mean, median, variance, kurtosis, skewness, and entropy, among others. Meanwhile, TR detection is approached by comparing shape similarity between thermal measurements in~\cite{Li:ISIE:2021}, independent component analysis in~\cite{Jung:Energies:2024}, and outlier analysis~\cite{Sun-Wang:TPE:2022,Yuan:JES:2023}. With the increasing data abundance for LiB systems, machine learning has stood out as a natural choice for data-driven TR detection, thanks to its capability of handling complex data and uncovering latent features or patterns~\cite{Li-Weng:JES:2023,Zhao:PECS:2024}. In this research direction, supervised learning has gained popularity, which uses labeled datasets to predict faults. The studies in~\cite{Ojo:TIE:2021,Li-Peng:TPE:2022,Zhao:RESS:2024} use different types of neural networks to predict the surface temperature of LiBs and then compare their predictions against measurements to pinpoint thermal anomalies. Neural networks can also be trained directly to learn TR-related patterns directly from the measurements of temperature, current, and voltage~\cite{Du:JES:2023,Lekoane:JES:2024,Liu:JES:2024}. In~\cite{Jia-Li:JPS:2022}, neural networks are trained to classify thermal safety/risk levels. Integrating physics with neural networks can further enhance the fault detection accuracy~\cite{Kim:JES:2023,Ke:AE:2025}, with the idea tracing to~\cite{Tu:AE:2023}. Riding on the advances in machine learning, other techniques---such as transfer learning, Gaussian processes, and support vector machines---have found their way into thermal fault prediction or detection~\cite{Chen-Xiong:AE:2018,Schaeffer:CRPS:2024,Masalkovaite:NatComm:2024}. Beyond supervised learning, unsupervised learning can also play a valuable role in TR detection. Instead of relying on labeled data for training, this approach extracts computable features from raw data and applies techniques like clustering, distance calculations, or hypothesis testing to these features to identify the onset or occurrence of TR or ISCs~\cite{Schmid:IEEE-SJ:2021,Qiao:AE:2022,Yu-Yang:JES:2025}. 

Model-based methods provide another powerful means for TR and ISC detection. Dynamic models encapsulate physics-based understanding of LiBs' behaviors within parsimonious mathematical structures, making them both expressive and generalizable. Supplementing these models with real-time operating data makes it possible to develop detection methods with high data efficiency, fast computation, and broad applicability. The work in~\cite{Dey:TCST:2019} employs a spatially distributed electro-thermal model with a thermal fault expression and builds observers to estimate the parameter indicative of the thermal fault. Some other studies have applied observers or Kalman filters to spatio-temporal thermal models to estimate and monitor the temperature distribution for LiBs~\cite{Kim:TCST:2014,Tian:TCST:2019,Zhang-Dey:AUTO:2021,Tu:CDC:2022}, which are extensible to TR detection. Lumped models are often favored for real-time TR detection, as their low-order equivalent circuit structures facilitate both implementation and computation. Such models usually couple the well-known Thevenin model with a thermal model. Based on them, one can evaluate the likelihood of TR by assessing the discrepancy between model predictions and measurements relative to a predefined threshold~\cite{Klink:JPS:2021}. As shown in~\cite{Dey:CEP:2016,Marcicki:DSCC:2010,Dong:JES:2021,Jia:JPS:2024}, fault observers and Kalman filters allow to choose and compute thresholds based on a LiB's immediate operating condition. While these studies focus on cell-level fault detection, pack-level detection has attracted growing interest. The method proposed in~\cite{Ouyang:JPS:2015} estimates ISC-related parameters for each cell and then assesses the consistency of these parameters within the pack to enable ISC detection. In~\cite{Bhaskar:AE:2025}, observers are designed to estimate the leakage current of individual cells, allowing cell-to-cell comparison to reveal potential ISC incidents. The effectiveness of fault detection often depends on the availability and placement of sensors, as shown by structural analysis in~\cite{Cheng:TTE:2022}. To meet the practical need of using a minimum number of sensors, the study in~\cite{Farakhor:ICPS:2024,Farakhor:ACC:2025} proposes to include structural properties inherent to LiBs packs as extra information into a moving horizon estimation strategy for fault detection. Such properties include the behavioral uniformity among the cells under normal operation and the natural sparsity of fault occurrence. 

\subsection{Statement of Contributions}\label{Sec:Contribution-Statement}

In this paper, we present a systematic study that encompasses: 1) dynamic modeling for ISC-induced TR, 2) early detection of TR events based on the proposed model. The specific contributions are as follows.

First, we develop the first equivalent circuit model to describe a LiB's electro-thermal behavior during an ISC-triggered TR event. As reviewed in Section~\ref{Sec:Review-TR-ISC}, such a model has been absent from the literature, which to date includes only equivalent circuit models for ISCs or 1D/3D TR models grounded in reaction kinetics. 
While having an equivalent circuit structure, the proposed model---named BattBee---draws inspiration from thermo-electrochemical modeling and is designed to capture multiple key phenomena, including the impacts of a growing ISC on voltage drop, exothermic heat generation, and the resulting rapid temperature rise. With its structural simplicity, the BattBee model is well-suited for efficient TR simulation and detection in real-time usage scenarios. 

Second, building on the BattBee model, we develop an ISC and TR detection approach leveraging the notion of observer-based fault detection. What makes our approach appealing is its simplicity and rigor in design, along with closed-form detection logics for the sake of practical implementation. To this end, we perform piecewise linearization of the BattBee model, converting it into a set of linear submodels corresponding to different operating regimes. Given these linear submodels, we design a linear observer-based fault detection framework and derive analytical bounds for the observer's residuals---the actual residuals exceeding these bounds imply an ISC or TR event to occur. By design, the proposed approach is both easy to implement and fast to compute. 

We validate both the BattBee model and the proposed fault detection approach using simulations and experimental data. The results demonstrate the effectiveness of the research results, highlighting their potential for a range of applications.

\subsection{Organization} \label{Organization}
The remainder of the paper is organized as follows. Section~\ref{Sec:BattBee-Model-Development} develops the BattBee model, showing its governing equations and introducing the rationale underpinning it. Section~\ref{sec:BattBee-Model-Validation} validates the BattBee model through simulations and experiments under both normal conditions and ISC-induced TR scenarios. Section~\ref{Sec:TR-ISC-Fault-Detection} presents the design of the ISC and TR detection approach, with the validation results reported in Section~\ref{Sec:TR-ISC-Fault-Detection-Validation}. Finally, Section~\ref{Sec:Conclusion} concludes the paper.

\section{Development of the BattBee Model} \label{Sec:BattBee-Model-Development}

In this section, we develop the BattBee model and elaborate the rationale behind its design. Briefly, BattBee employs   physics-inspired design and integrates multiple equivalent circuits to capture the coupled electrical and thermal behaviors under ISC conditions. 

\begin{figure} [t]
 \begin{center}
 \includegraphics[width=0.46\textwidth]{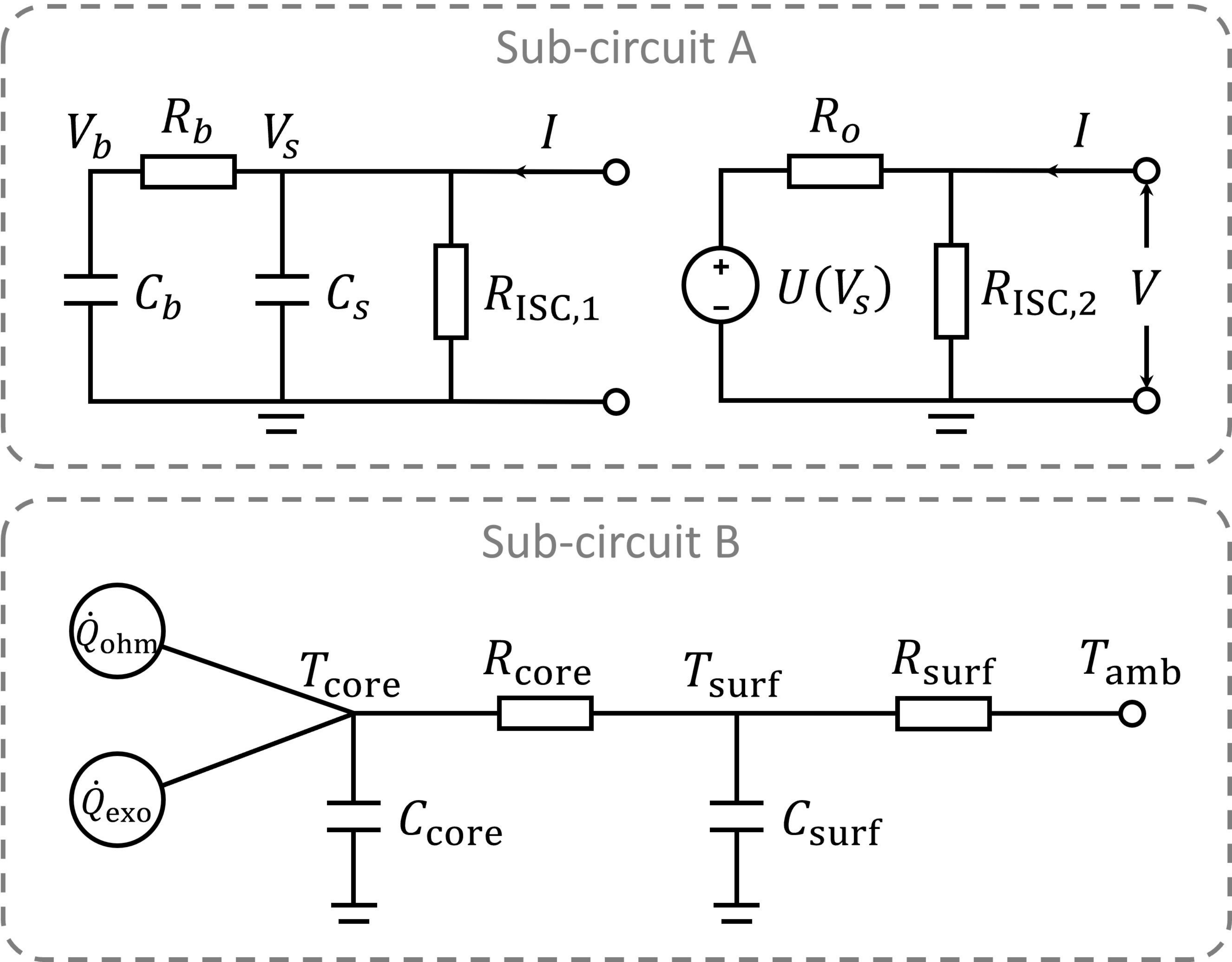}\\
 \caption{The BattBee model's equivalent circuit structure.}\label{Fig:BattBee}
 \end{center}
\end{figure}

\subsection{The BattBee Model} \label{BattBee-Model-Governing-Equations}

As shown in Fig.~\ref{Fig:BattBee}, the BattBee model includes two sub-circuits, labeled as A and B, to simulate the electric and thermal behaviors, respectively, when an ISC happens. The following is a detailed explanation for each. 

 Sub-circuit A extends the nonlinear double-capacitor model from~\cite{Tian:TCST:2021} by incorporating two ISC resistors, $R_{\mathrm{ISC},1}$ and $R_{\mathrm{ISC},2}$. Nominally, the left part of the circuit approximates the charge transfer dynamics resulting from lithium-ion diffusion in an electrode, and the right part represents the terminal voltage. Under normal operating conditions, $R_{\mathrm{ISC},1}$ and $R_{\mathrm{ISC},2}$ are effectively infinite. However, in the event of an ISC, their values decrease significantly---approaching zero in the case of a complete ISC. As an interesting third case, though beyond the scope of this paper, the self-discharge effect is tantamount to setting $R_{\mathrm{ISC},1}$ to a large but finite value, while keeping $R_{\mathrm{ISC},2}$ infinite. 
 
The governing equations of sub-circuit A are given by
\begin{subequations}\label{BattBee-Electric-Governing-Equations}
\begin{align}
\label{eqn:Vb_dynamics}
{\dot{V}}_b (t) &= \frac{V_s (t) - V_b (t)}{R_bC_b} , \\ 
\label{eqn:Vs_dynamics}
{\dot{V}}_s (t) &=\frac{V_b (t) - V_s (t)}{R_bC_s} - \frac{ V_s(t)}{R_{\mathrm{ISC,1}}C_s} +\frac{I(t)}{C_s},\\
\label{eqn:V_out_dynamics}
V (t) &= \frac{1}{1+\frac{R_o}{R_\mathrm{ISC,2}}} \left[U(V_s (t))+I(t)R_o\right] .
\end{align}
\end{subequations}
\iffalse
\begin{subequations}\label{BattBee-Governing-Equations}
\begin{align}
\label{eqn:Vb_dynamics}
{\dot{V}}_b (t) &=-\frac{1}{R_bC_b}V_b (t) +\frac{1}{R_bC_b}V_s (t), \\ 
\label{eqn:Vs_dynamics}
{\dot{V}}_s (t) &=\frac{1}{R_bC_s}V_b(t)-\left( \frac{1}{R_bC_s}+\frac{1}{R_{\mathrm{ISC,1}}C_s}\right) V_s(t)+\frac{1}{C_s}I(t),\\
\label{eqn:V_out_dynamics}
V (t) &=\left[U(V_s (t))+I(t)R_o\right] \frac{1}{1+\frac{R_o}{R_\mathrm{ISC,2}}}.
\end{align}
\end{subequations}
\fi
In above, the capacitors $C_b$ and $C_s$ store charge; $V_b$ and $V_s$ are the voltages across them; $I$ is the applied current, with $I>0$ for charging and $I<0$ for discharging; $R_b$ is the charge transfer resistance, and $R_o$ is the internal resistance; $U(\cdot)$ is the open-circuit voltage (OCV) source. We normalize $V_b$ and $V_s$ such that $0\leq V_b, V_s \leq 1$ for convenience. As such, the LiB's capacity is $C_b+C_s$, and the state-of-charge (SoC) is
\begin{align}\label{eqn:SoC_dynamics}
\mathrm{SoC} (t) =\frac{C_bV_b\left(t\right)+C_sV_s\left(t\right)}{C_b+C_s} \times 100\%.
\end{align}
The LiB is fully charged with $\mathrm{SoC} = 100\%$ when $V_b=V_s=1$, and fully depleted with $\mathrm{SoC} = 0 \%$ when $V_b=V_s=0$. Note that, at equilibrium state under normal operation, $V_b=V_s$, and thus $\mathrm{SoC}=V_s$. Without loss of generality, we can parameterize $U(\cdot)$ as a polynomial of $V_s$:
\begin{equation}\label{eqn:OCV_dynamics}
U(V_s) = \sum_{i=0}^{N} \lambda_i V_s^i,
\end{equation}
where $\lambda_i$ for $i=1,\ldots,N$ are the coefficients, and $N$ is the polynomial's order. 

When an ISC occurs, the two electrodes of a LiB cell, normally kept apart by the separator, come into direct contact. This would result in high internal currents, fast charge draining, and sharp voltage drop. These phenomena can be grasped by the introduction of $R_{\mathrm{ISC},1}$ and $R_{\mathrm{ISC},2}$, as described in Eq.~\eqref{BattBee-Electric-Governing-Equations}. Specifically, when $R_{\mathrm{ISC},1}$ decreases from infinity to a small value, the charge stored in $C_b$ and $C_s$ is quickly discharged through $R_{\mathrm{ISC},1}$, producing a large internal current. As a consequence of this rapid charge loss, $V_b$, $V_s$ and $\mathrm{SoC}$ all decline, causing a drop in the terminal voltage $V$. Further, when $R_{\mathrm{ISC},2}$ also becomes small, $V$ would fall precipitously. Note that the severity of the ISC event determines the values of $R_{\mathrm{ISC},1}$ and $R_{\mathrm{ISC},2}$. Compared to this proposed design, existing equivalent circuit models for ISCs in the literature, e.g.,~\cite{Ouyang:JPS:2015,Sun:JPS:2024,Lai:ElecActa:2018,Bhaskar:AE:2025}, are limited to capturing only the voltage drop.

Sub-circuit B approximates the LiB's thermal dynamics when an ISC triggers a TR event. {\color{blue}It extends a lumped thermal model proposed in~\cite{Lin:JPS:2014}, which  simplifies the spatial temperature distribution of a cylindrical cell into a heat flow between the cell’s core and surface. However, different from~\cite{Lin:JPS:2014}, sub-circuit B  accounts for various types of heat generation that dominate the TR process, with particular emphasis on heating due to exothermic reactions.} 
Sub-circuit B is governed by
\begin{subequations}
\begin{align}
\label{eqn:T_core_dynamics}
{\dot{T}}_\mathrm{core} (t) &=
\frac{T_\mathrm{surf} (t) - T_\mathrm{core}(t)}{R_\mathrm{core}C_\mathrm{core}}+ \frac{\dot Q(t)}{C_\mathrm{core}}, \\
\label{eqn:T_surf_dynamics}
{\dot{T}}_\mathrm{surf} (t) &=\frac{T_\mathrm{core} (t) -T_\mathrm{surf} (t) }{R_\mathrm{core}C_\mathrm{surf}}-\frac{T_\mathrm{surf} (t) -T_\mathrm{amb}}{R_\mathrm{surf} C_\mathrm{surf}},
\end{align}
\end{subequations}
where $T_\mathrm{core/surf}$, $C_\mathrm{core/surf}$ and $R_\mathrm{core/surf}$ are the temperatures, thermal capacitances, conductive/convective resistances at the LiB's core and surface, respectively, and $T_\mathrm{amb}$ is the ambient temperature.

Further, the heat source term $\dot Q$ sums up the ohmic heat rate $\dot{Q}_\mathrm{ohm} (t)$ and the exothermic heat rate $\dot{Q}_\mathrm{exo} (t)$:
\begin{align}
\dot Q(t) =  \dot{Q}_\mathrm{ohm} (t) +\dot{Q}_\mathrm{exo} (t),
\end{align}
where
\begin{align}\label{eqn:Q_ohm_dynamics}
{\dot{Q}}_\mathrm{ohm}\left(t\right)=I^2(t) R_o.
\end{align}
Further, $\dot{Q}_\mathrm{exo} (t)$ accounts for the large amounts of heat released in exothermic reactions that take place during the TR process. It includes $\dot Q_{\mathrm{ec}} (t)$, the heat rate due to electrochemical reactions, and $\dot Q_{\mathrm{decomp}} (t)$, the heat rate due to the decomposition of active materials: 
\begin{align}
\dot{Q}_\mathrm{exo} (t) = \dot Q_{\mathrm{ec}} (t) + \dot Q_{\mathrm{decomp}}(t). 
\end{align}
According to~\cite{Coman:JES:2017}, we have
\begin{align}
\label{eqn:Q_ISC_dynamics}
\dot Q_{\mathrm{ec}} (t)  = - h_\mathrm{ec} \cdot \dot{\mathrm{SoC}}(t)=-h_\mathrm{ec} \frac{V_s\left(t\right)+I\left(t\right)R_\mathrm{ISC,1}}{R_\mathrm{ISC,1}\left(C_b+C_s\right)},
\end{align}
where $h_\mathrm{ec}$ is an enthalpy-related coefficient for heating caused by the ISC-induced SoC loss. The active material decomposition occurs to the SEI, electrodes, and electrolyte when the LiB's temperature becomes high. This is a dynamic process resulting in complex heat generation. To capture this behavior while maintaining the model's simplicity, we introduce a lumped description for $ \dot{Q}_\mathrm{decomp}$: 
\begin{align}\label{eqn:Q_decomp_dynamics} 
\dot{Q}_{\mathrm{decomp}}(t) = \frac{\alpha_{1} \exp\left[\alpha_{2} \left(T_\mathrm{core} (t) - T_\mathrm{onset}\right)\right]}{1 + \alpha_3 \exp\left[\alpha_{4} \left(T_\mathrm{core} (t) - T_\mathrm{onset}\right)\right]},
\end{align}
where $T_\mathrm{onset}$ is the TR onset temperature, and $\alpha_i$ for $i=1,\ldots,4$ are coefficients. These coefficients are chosen such that $\dot{Q}_{\mathrm{decomp}}$ is small or negligible when $T_\mathrm{core} < T_\mathrm{onset}$, but large when $T_\mathrm{core} > T_\mathrm{onset}$. Also, $\dot{Q}_{\mathrm{decomp}}$ would take a negligible value after $T_\mathrm{core}$ or $T_\mathrm{surf}$ reaches a pre-specified peak temperature, due to the active material depletion. The expression for $\dot{Q}_{\mathrm{decomp}}$ in Eq.~\eqref{eqn:Q_decomp_dynamics} represents the best formulation we identified after exploring various alternatives. Nonetheless, other valid forms may also exist. 

Finally, $R_\mathrm{surf}$ varies with $T_\mathrm{surf} - T_\mathrm{amb}$, due to the dependence of convection and radiation on temperature gradients~\cite{bergman2011fundamentals}. To reflect this, we can parameterize $R_\mathrm{surf}$ into a linear form as
\begin{equation}\label{eqn:R_surf_dynamics}
{R_\mathrm{surf}\left( T_\mathrm{surf}\left(t\right)\right)} = R_\mathrm{surf,0} \cdot \left( 1 - \beta \left( T_\mathrm{surf}\left(t\right) - T_\mathrm{amb} \right) \right),
\end{equation}
where $R_\mathrm{surf,0}$ is what $R_\mathrm{surf}$ is under zero temperature gradient, and with $\beta$ is a coefficient typically taking $1/600$.

Putting together Eqs.~\eqref{BattBee-Electric-Governing-Equations}-\eqref{eqn:Q_decomp_dynamics}, we obtain the mathematical setup of the BattBee model. As highlighted in Section~\ref{Sec:Contribution-Statement}, this model is the first equivalent circuit model to describe the electro-thermal behavior of a LiB under an ISC-triggered TR event. By design, it presents a concise, low-order structure while   capturing various key phenomena underlying TR. This combination of simplicity and fidelity makes the BattBee model competent for TR detection in real world. 

\subsection{Rationale for the BattBee Model}

Let us elaborate on the construction of the BattBee model. Note that sub-circuit B builds upon the formulations presented in~\cite{Lin:JPS:2014,Coman:JES:2017} and has been explained in Section~\ref{BattBee-Model-Governing-Equations}. Therefore, we focus solely on justifying the design of sub-circuit A. We begin with proposing an ISC-including single particle model (SPM and SPMw/ISC), drawing upon the pseudo-2D ISC model in~\cite{Kong:JES:2020}. We then reduce the SPMw/ISC model to an equivalent circuit representation, which is further simplified to yield sub-circuit A of the BattBee model.

\begin{figure} [t]
 \begin{center}
 \includegraphics[width=0.46\textwidth]{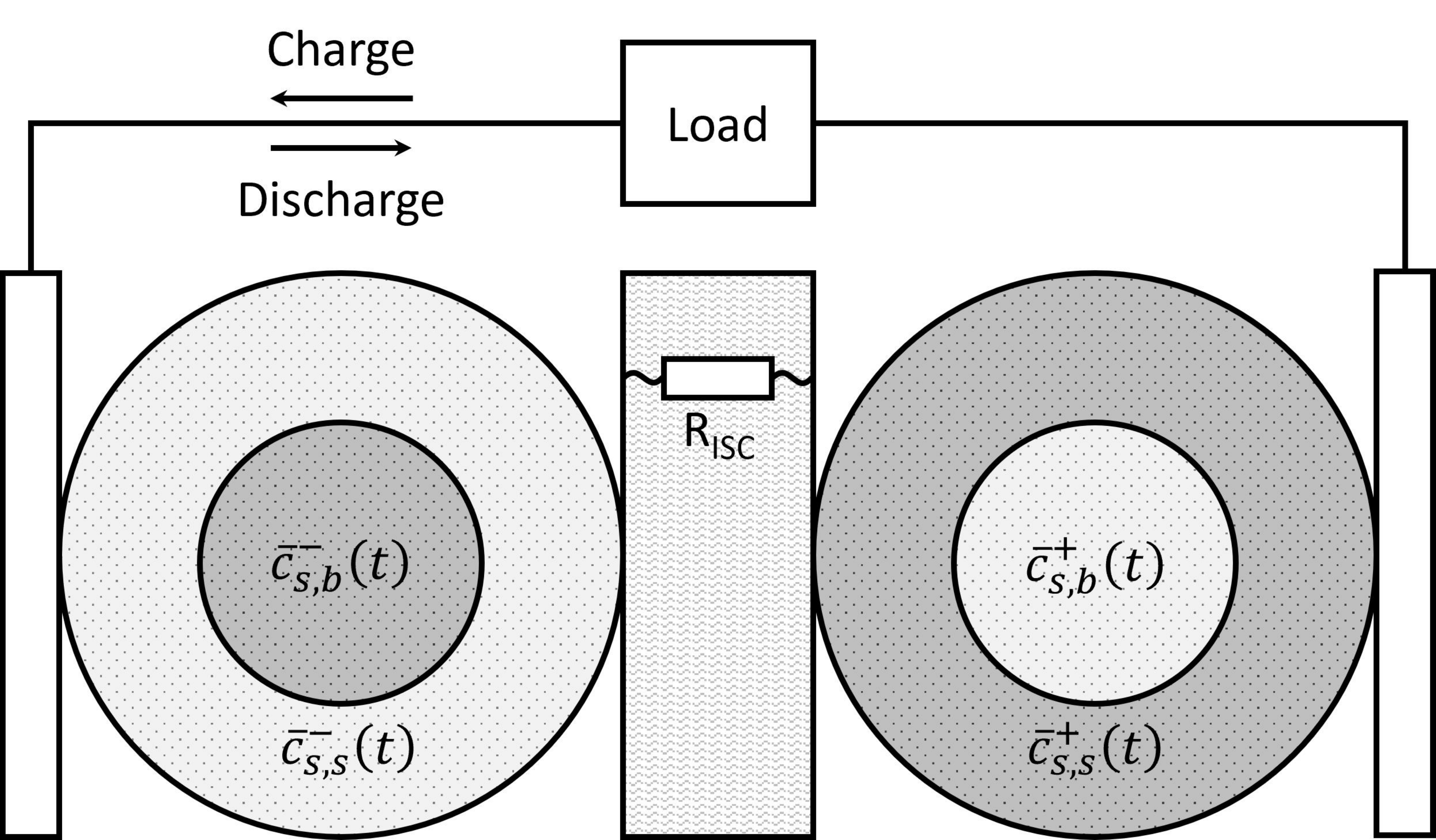}\\
 \caption{Single particle model with an ISC resistance.}
 \label{Fig:SPMwISC_Diagram}
 \end{center}
\end{figure}

Well-known in the literature, the SPM model idealizes the electrodes of a LiB cells as spherical particles while neglecting the electrolyte dynamics~\cite{Moura:TCST:2017}. Lithium ions move within the spherical particles following Fick's second law of diffusion: 
\begin{align}\label{eqn:c_Li_SPM_dynamics}
\frac{\partial c_s^\pm (r,t)}{\partial t}=\frac{1}{r^2}\frac{\partial}{\partial r}\left[D_s^\pm r^2\frac{\partial c_s^\pm (r,t)}{\partial r} \right],
\end{align}
where $c_s^\pm$ is the lithium-ion concentration within the solid-phase particle, $r$ is the radial coordinate, $D_s^\pm$ is the diffusion coefficient, and the superscripts $+$ and $-$ denote the positive and negative electrodes, respectively. The boundary conditions for Eq.~\eqref{eqn:c_Li_SPM_dynamics} are 
\begin{equation}\label{eqn:c_Li_SPM_boundary}
\left. \frac{\partial c_s^\pm}{\partial r}\right|_{r=0}=0,\ \ \ 
\left.\frac{\partial c_s^\pm}{\partial r}\right|_{r= r_s^\pm} = -\frac{1}{D_s^\pm}j^\pm\left(t\right),
\end{equation}
where $r_s^\pm$ is the particle's radius for each electrode, and $j^\pm$ is the molar ion flux. 
Given an ISC, electrons can pass through the separator during charging/discharging. To make an analogy to this, an ISC resistance, $\mathsf{R}_{\mathsf{ISC}}$, is introduced to the separator to conduct the electrons, causing an ISC-induced current, $\mathsf{I}_{\mathsf{ISC}}$, as shown in Fig.~\ref{Fig:SPMwISC_Diagram}. By Ohm's law, 
\begin{align}\label{eqn:I_ISC_SPM}
\mathsf{I}_\mathsf{ISC}(t)= \frac{ \Delta \mathsf{U}_{\mathsf{sep}}}{\mathsf{R}_\mathsf{ISC}},
\end{align}
where $\Delta \mathsf{U}_{\mathsf{sep}}$ is the voltage across the separator. Note that $\mathsf{R}_\mathsf{ISC} = \infty$ in the ISC-free case. Due to the presence of $\mathsf{I}_\mathsf{ISC}$, we have
\begin{align}\label{eqn:molar_flux}
j^{\pm}(t)= \pm \frac{I(t) - \mathsf{I}_\mathsf{ISC}(t)}{Fa^\pm L^\pm S},
\end{align}
where $I$ is the applied current, $F$ is Faraday’s constant, $a^\pm$ is the specific interfacial surface area, $L^\pm$ is the electrode thickness, and $S$ is the cell’s cross-sectional area.  
The Butler-Volmer equation mandates that
\begin{align}\label{Butler-Volmer}
j^\pm(t) = \frac{i_0^\pm}{F}\left[ \exp \left( \frac{\alpha F}{\mathcal{R}T} \eta^\pm (t) \right) - \exp\left( - \frac{\alpha F}{\mathcal{R}T} \eta^\pm (t) \right) \right], 
\end{align}
where $i_0^\pm$ is the exchange current density, $\alpha$ is the charge transport coefficient, $\mathcal{R}$ is the universal gas constant, and $\eta^\pm$ is the overpotential of reaction. Specifically, $\eta^\pm$ is given by
\begin{align} \label{eta-combined-sum}
\eta^\pm(t) = \phi_s^\pm(t) - \phi_e^\pm(t) - U^\pm(c^\pm_{ss} (t)) - FR_f^\pm j^\pm(t),
\end{align}
where $\phi_s^\pm$ and $\phi_e^\pm$ are the potential in the solid-phase electrodes and electrolyte, respectively, $U^\pm (\cdot)$ is the open-circuit potential, $R_f^\pm$ is the film resistance, and $c_{ss}^\pm\left(t\right) = c_s^\pm(r_s^\pm,t)$. By Eq.~\eqref{Butler-Volmer}, we have
\begin{align} \label{eta-sinh-inv}
\eta^\pm(t) = \frac{\mathcal{R} T}{\alpha F} \sinh^{-1}\left(\pm \frac{I(t) - \mathsf{I}_\mathsf{ISC}(t)}{2i_0^\pm a^\pm L^\pm S }\right).
\end{align}
Here, $\sinh^{-1}(\cdot)$ is the inverse hyperbolic sine function. From the above, the terminal voltage $V$ then is
\begin{align}\label{eqn:V_out_SPM_dynamics} \nonumber
V(t) &= \phi_s^+(t) - \phi_s^-(t) \\ \nonumber
&= U^+(c^+_{ss} (t))-U^-(c^-_{ss} (t)) - \Delta \mathsf{U}_{\mathsf{sep}} + \phi_e^+(t) - \phi_e^-(t)\\
& \quad+\eta^+(t) - \eta^-(t) + F R_f^+ j^+(t) - F R_f^- j^- (t).  
\end{align}
Collecting Eqs.~\eqref{eqn:c_Li_SPM_dynamics}-\eqref{eqn:V_out_SPM_dynamics}, we develop the SPMw/ISC model, which is the SPM model incorporating an ISC mechanism. Next, we simplify this model towards the BattBee model

The simplification follows similar lines in our prior studies in~\cite{Fang:TCST:2017, Tian:TCST:2021, Biju:AE:2023}. We begin by subdividing each spherical particle into two continuous finite volume elements. The first element is a ball with a radius $r_b^\pm$, which is an electrode's bulky inner part; the second element is a hollow sphere with an outer radius of $r_s^\pm$ and an inner radius of $r_b^\pm$, which is the electrode's shell. Based on Eq.~\eqref{eqn:c_Li_SPM_dynamics}, the amounts of lithium ions within each element, $n_{s,b}^\pm$ and $n_{s,s}^\pm$, evolve by 
\begin{subequations}\label{ns-nb-dynamics}
\begin{align}\label{eqn:n_Li_b_dynamics}
 {\dot{n}}_{s,b}^{\pm}(t)
&=\int_{0}^{r_{b}^{\pm}}{\frac{\partial c_s^{\pm}(r,t)}{\partial t} 4\pi r^2dr} = D_s^{\pm} S_b^\pm \left. \frac{\partial c_s^{\pm}(r,t)}{\partial r}\right|_{r_{b}^\pm},\\\label{eqn:n_Li_s_dynamics} \nonumber
{\dot{n}}_{s,s}^{\pm}(t)
&=\int_{r_{b}^{\pm}}^{r_s^{\pm}}{\frac{\partial c_s^{\pm} (r,t)}{\partial t}  4\pi r^2dr}\\
&= D_s^{\pm} S_s^\pm \left. \frac{\partial c_s^{\pm} (r,t)}{\partial r}\right|_{r_s^{\pm}} - D_s^{\pm}S_b^\pm \left. \frac{\partial c_s^{\pm}(r,t)}{\partial r}\right|_{r_{b}^{\pm}},
\end{align}
\end{subequations}
where $S_b^\pm = 4\pi \left( r_{b}^{\pm} \right)^2 $ and $S_s^\pm = 4\pi \left( r_{s}^{\pm} \right)^2 $.
To proceed, we consider the average lithium-ion concentration within the bulky inner part   and that within the shell, which are defined as
\begin{align*}
\bar c_{s,b}^\pm(t) &= \frac{n_{s,b}^\pm (t)}{\Delta v_b^\pm}, \ \ \ \
\bar c_{s,s}^\pm(t) = \frac{n_{s,s}^\pm(t)}{\Delta v_s^\pm}, \\
\Delta v_b^\pm &= {4\pi \left(r_b^\pm\right)^3 \over 3}, \ \ \ \Delta v_s^\pm = {4\pi \left[ \left(r_s^\pm\right)^3 - \left(r_b^\pm\right)^3) \right] \over 3}.
\end{align*}
Then, it follows from Eq.~\eqref{ns-nb-dynamics} that 
\begin{subequations}\label{eqn:c_Li_b_dynamics}
\begin{align}
{\dot{\bar{c}}}_{s,b}^{\pm}(t)&=\frac{ D_s^{\pm} S_b^\pm }{\Delta v_b^\pm} \left. \frac{\partial c_s^{\pm}(r,t)}{\partial r}\right|_{r_{b}^\pm},\\ \label{eqn:c_Li_s_dynamics}
{\dot{\bar{c}}}_{s,s}^{\pm}(t)
&=\frac{D_s^\pm S_s^\pm }{\Delta v_s^\pm } \left. \frac{\partial c_s^{\pm} (r,t)}{\partial r}\right|_{r_s^{\pm}} - \frac{D_s^\pm S_b^\pm }{\Delta v_s^\pm } \left. \frac{\partial c_s^{\pm}(r,t)}{\partial r}\right|_{r_{b}^{\pm}}. 
\end{align}
\end{subequations}
Moving forward, we approximate the concentration gradient by 
\begin{align*}
\left. \frac{\partial c_s^{\pm}(r,t)}{\partial r}\right|_{r_{b}^\pm} \approx \frac{\bar{c}_{s,s}^{\pm}\left(t\right)-\bar{c}_{s,b}^{\pm}\left(t\right)}{\frac{r_s^{\pm}}{2}}. 
\end{align*}
This approximation, along with Eq.~\eqref{eqn:c_Li_SPM_boundary} and Eq.~\eqref{eqn:molar_flux}, leads to 
\begin{subequations}\label{c-bar-dynamics-final-form}
\begin{align}
 \dot{\bar{c}}_{s,b}^{\pm}(t) =& \frac{ 2 D_s^{\pm} S_b^\pm }{\Delta v_b^\pm r_s^\pm} 
\left(\bar{c}_{s,s}^{\pm}\left(t\right)-\bar{c}_{s,b}^{\pm}\left(t\right)\right),\\ \nonumber
 \dot{\bar{c}}_{s,s}^{\pm}(t) =& - \frac{ 2 D_s^{\pm} S_b^\pm }{\Delta v_s^\pm r_s^\pm} 
\left(\bar{c}_{s,s}^{\pm}\left(t\right)-\bar{c}_{s,b}^{\pm}\left(t\right)\right) \\ 
& \mp \frac{ S_s^\pm }{\Delta v_s^\pm F a^\pm L^\pm S}
\left[ I(t) - \mathsf{I}_\mathsf{ISC}(t) \right].
\end{align}
\end{subequations}
A close similarity between Eq.~\eqref{c-bar-dynamics-final-form} and Eqs.~\eqref{eqn:Vb_dynamics}-\eqref{eqn:Vs_dynamics} emerges upon comparison, with the following observations: 
\begin{itemize}
\item $V_b$ and $V_s$ correspond to $ {\bar{c}}_{s,b}^{\pm}$ and ${\bar{c}}_{s,s}^{\pm}$;

\item $C_b$ and $C_s$ mirror $\Delta v_b^\pm$ and $\Delta v_s^\pm$, respectively, which is reasonable since $C_b$ and $C_s$ represents the capacitance for charge storage within the BattBee model, while $\Delta v_b^\pm$ and $\Delta v_s^\pm$ denotes the volumetric space for lithium-ion storage in the SPMw/ISC model; 

\item $R_b$ is related to $r_s^\pm /D_s^\pm S_b^\pm$, reflecting that the resistance for lithium-ion diffusion is proportional to $r_s^\pm$ and inversely proportional to $D_s^\pm$ and $S_b^\pm$;

\item $R_{\mathrm{ISC},1}$ is interpretable as resulting from $\mathsf{R}_\mathsf{ISC}$, given Eq.~\eqref{eqn:I_ISC_SPM}. 
\end{itemize}

Now, let us consider Eq.~\eqref{eqn:V_out_SPM_dynamics}. Since $\sinh^{-1}(\cdot)$ lends itself to linearization, we   approximately have    $\eta^\pm(t) \approx \pm R_\eta^{\pm} \left[ I(t) - \mathsf{I}_\mathsf{ISC}(t) \right]$, where $ R_\eta^{\pm}$ results from linearizing $\sinh^{-1}(\cdot)$ and can be viewed as a resistance term. Further, we assume $\phi_e^\pm(t) = R_e^\pm I(t)$, with $R_e^\pm$ denoting the electrolyte resistance. Similarly, the term $F R_f^+ j^+(t) - F R_f^- j^- (t)$ can be approximated as $R_{\mathrm{film}} I(t)$, where $R_{\mathrm{film}}$ represents an aggregated film resistance. With these approximations,~\eqref{eqn:V_out_SPM_dynamics} can be rewritten as
\begin{align}\label{eqn:V_out_SPM_dynamics-rewrite} \nonumber
V(t) &=  U^+(c^+_{ss} (t)) - U^-(c^-_{ss} (t)) - \mathsf{R}_\mathsf{ISC} \mathsf{I}_\mathsf{ISC}(t) \\&\quad + R_{\mathrm{total}} I(t)- R_\eta \mathsf{I}_{\mathsf{ISC}}(t),
\end{align}
where $R_\eta = R_\eta^+ - R_\eta^-$ and $R_{\mathrm{total}} = R_\eta + R_e^+ - R_e^- + R_{\mathrm{film}}$. For Eq.~\eqref{eqn:V_out_dynamics}, we can rewrite it as 
\begin{align}\label{V_out-rewrite} \nonumber
V(t) &= U(V_s (t))+ R_o I(t) -R_{\mathrm{ISC},2} I_{\mathrm{ISC},2} (t) \\ & \quad- \left( R_o - R_{\mathrm{ISC},2} \right) I_{\mathrm{ISC},2}(t),
\end{align}
where $I_{\mathrm{ISC},2}$ is the current flowing through $R_{\mathrm{ISC},2}$:
\begin{align*}
I_{\mathrm{ISC},2} (t) = \frac{U(V_s (t))+ R_o I(t) }{R_o+R_{\mathrm{ISC},2} }.
\end{align*}
Comparing Eq.~\eqref{V_out-rewrite} with Eq.~\eqref{eqn:V_out_SPM_dynamics-rewrite}, we can find out the following:
\begin{itemize}
  \item $U(\cdot)$ corresponds to $U^+(\cdot) - U^-(\cdot)$; 
  
  \item $R_{\mathrm{ISC},2}$ plays a role analogous to that of $\mathsf{R}_\mathsf{ISC}$ in the SPMw/ISC model; 
  
  \item $R_o$ corresponds to $R_{\mathrm{total}}$.
\end{itemize}

The above derivation shows that the BattBee model is a simplified version of the SPMw/ISC model. As such, the BattBee model is justifiable to capture the key dynamics associated with an ISC event, while its low-order circuit structure facilitates real-world model execution. This model can be improved further for better accuracy---for example, we can let $R_o$ vary with SoC and temperature. However, the current formulation is accurate enough for the purpose of TR and ISC detection and thus sufficient here, though exploring potential enhancements will be part of our future work.

\section{Validation of the BattBee Model} \label{sec:BattBee-Model-Validation}
In this section, we validate the BattBee model to demonstrate its effectiveness. The first validation is based on simulations using GT-SUITE/GT-AutoLion, a multi-physics simulation tool for LiBs; the second builds on experimental data, partly generated by Sandia National Laboratories~\cite{Lin:JES:2023} and partly collected at the University of Kansas.

\subsection{Simulation-Based Validation} \label{sec:BattBee-Model-Validation-Sim}
 
For the simulation, we consider an NCM811 cell paired with a graphite anode and employing an ethylene carbonate-ethyl methyl carbonate electrolyte. The cell is enclosed in an aluminum casing, with a nominal capacity of $25$ Ah and physical dimensions of $309\times 117 \times 5.9 $ mm. The cell is configured and simulated using GT-AutoLion.

We begin by identifying the BattBee model under normal operating conditions, where $R_{\mathrm{ISC},1}, R_{\mathrm{ISC},2} = \infty$ and $\dot Q_{\mathrm{exo}} = 0$. To this end, we perform charge/discharge simulations using a range of current profiles, including UDDS, LA92, SC09, and pulse profiles, under an assumed ambient temperature of $T_{\mathrm{amb}} = 25^\circ\mathrm{C}$. Underlying the simulations is a pseudo-2D physicochemical model embedded in GT-AutoLion~\cite{doyle:JES:1993}.

\begin{table*}[ht!]
\centering
\caption{Parameter values of the BattBee model identified from simulation.}
\begin{tabularx}{\textwidth}{X X X X X X X X}
\toprule
$C_b$ [F] & $C_s$ [F] & $R_b$ [$\Omega$] & $R_o$ [$\Omega$] & $C_\mathrm{core}$ [J/K] & $C_\mathrm{surf}$ [J/K] & $R_\mathrm{core}$ [K/W] & $R_\mathrm{surf,0}$  [K/W] \\ 
\midrule
76900.887 & 8115.772 & 1.236$\times10^{-3}$ & 4.322$\times10^{-3}$ & 162.760 & 168.129 & 0.020 & 3.865 \\
\bottomrule
\end{tabularx}
\label{table:BattBee_Model_Parameter_Sim}
\end{table*}
From the simulation data, we extract the full SoC-OCV relationship represented by $U(\cdot)$  and estimate the remaining parameters of the BattBee model using a Bayesian-optimization-based approach proposed in~\cite{Tu:ACC:2024}. The estimated parameter values are summarized in Table~\ref{table:BattBee_Model_Parameter_Sim}.

\begin{figure} [t]
 \begin{center}
 \includegraphics[width=0.46\textwidth]{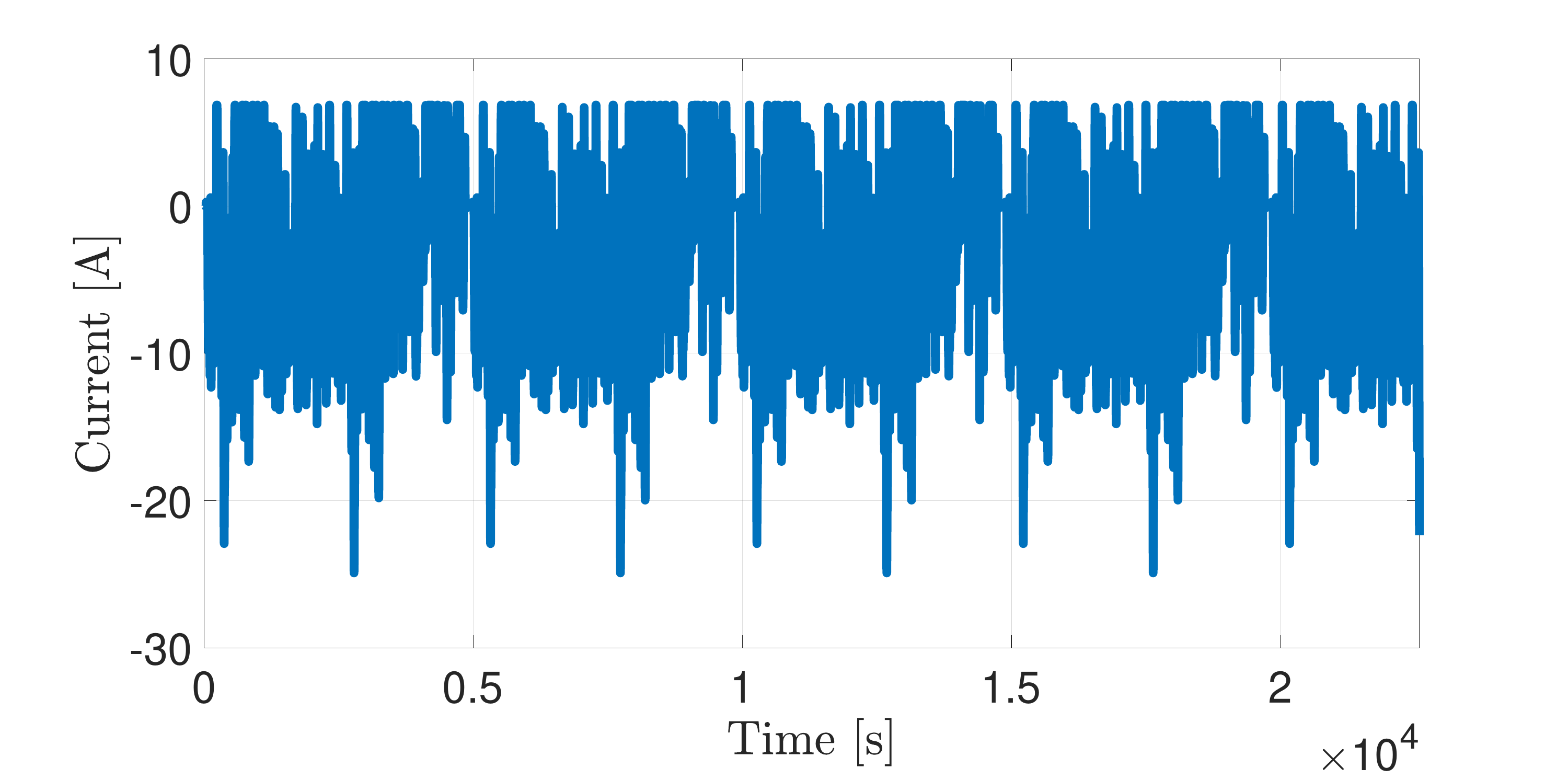}\\
 \caption{UDDS current profile.} \label{Fig:Sim_Val_Nom_I}
 \end{center}
\end{figure}
\begin{figure} [t]
 \begin{center}
 \includegraphics[width=0.46\textwidth]{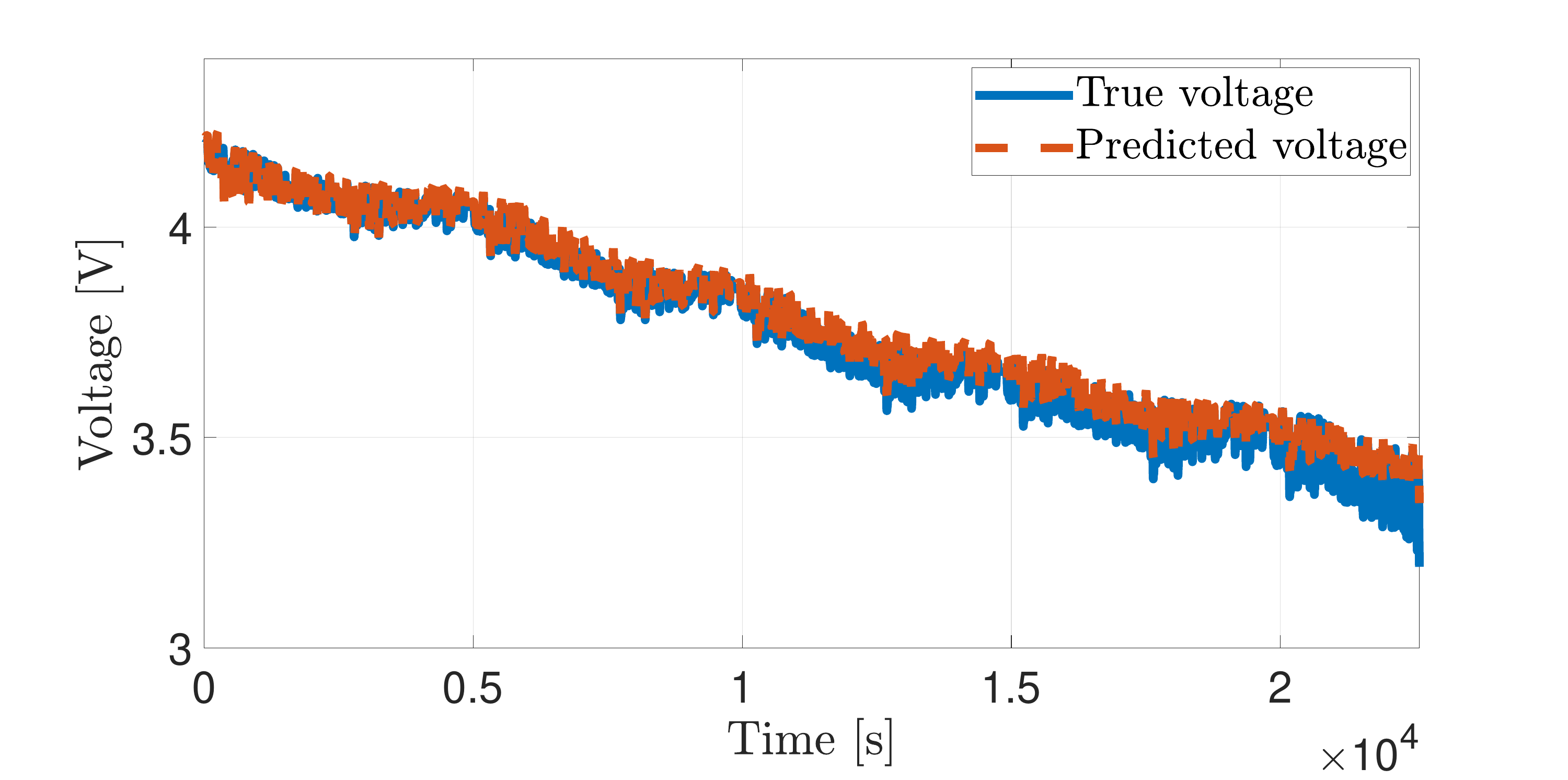}\\
 \subcaption[]{BattBee’s voltage prediction versus measurement under normal conditions.}
 \label{fig:Sim_Val_Nom_V}
 \includegraphics[width=0.46\textwidth]{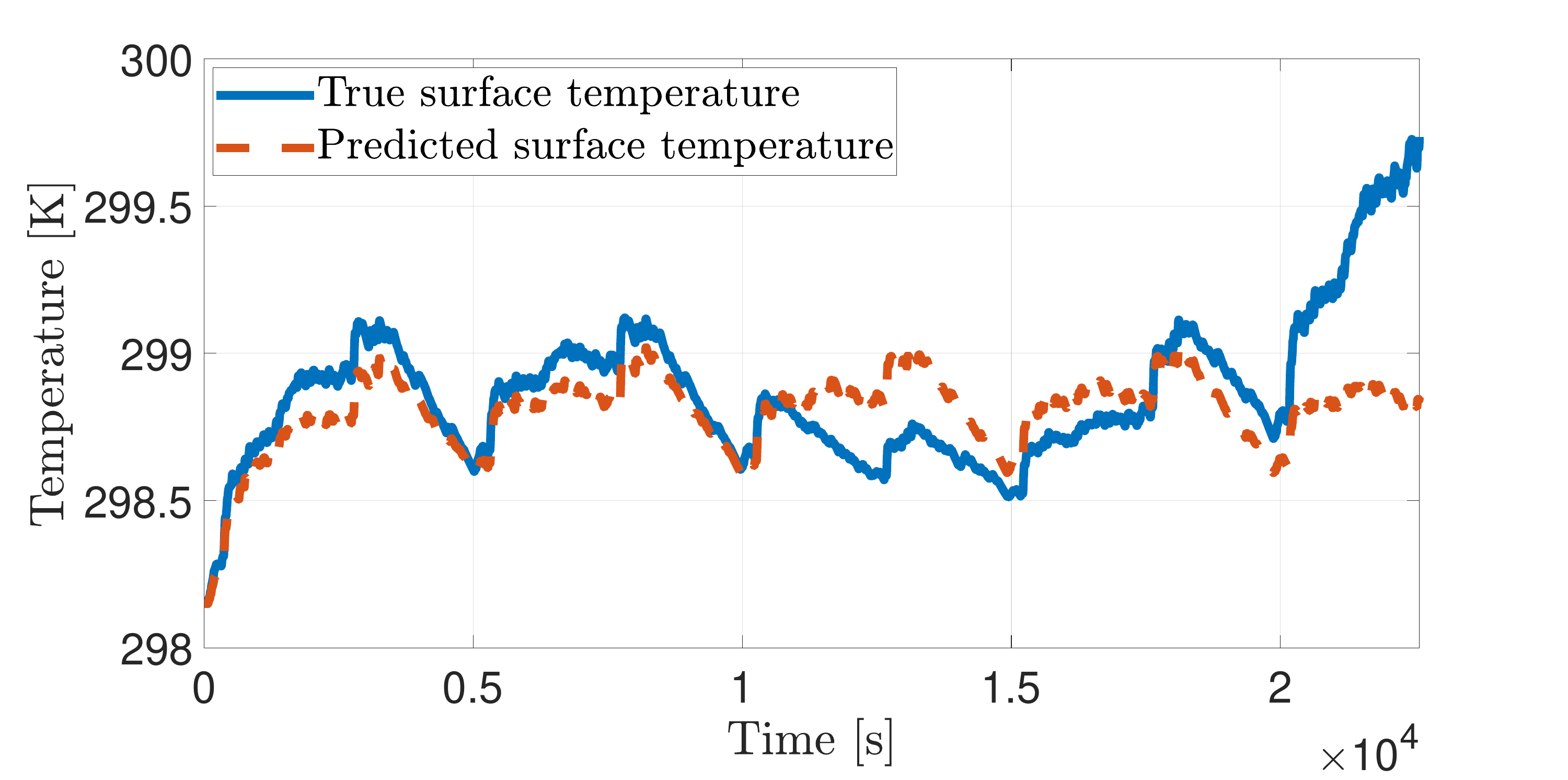}\\
 \subcaption[]{BattBee’s temperature prediction versus measurement under normal conditions.}
 \label{fig:Sim_Val_Nom_T}
 \end{center}
\caption{Simulation-based validation of the BattBee model under normal operating conditions.}
\label{Fig:Sim_Val_Nom}
\end{figure}
To evaluate predictive accuracy, we compare the BattBee model's voltage and temperature predictions with GT-AutoLion simulation results under a UDDS profile scaled to peak at 1 C (see Fig.~\ref{Fig:Sim_Val_Nom_I}). Fig.~\ref{Fig:Sim_Val_Nom} presents this comparison. Overall, the BattBee model demonstrates good agreement with the reference data. The average root mean squared error (RMSE) for voltage prediction is $33$ mV, and that for surface temperature prediction is $0.22$ K. Some discrepancies are observed at low SoC, resulting from the use of a constant $R_o$. We have adopted the choice of a fixed $R_o$ so that the BattBee model strikes a balance between model simplicity and accuracy. However, despite this simplification, the resulting level of accuracy still remains well within acceptable bounds for the purpose of TR and ISC detection.  

We proceed forward to evaluate the BattBee model under ISC-induced TR conditions. To this end, we first simulate an ISC event within GT-AutoLion-3D by introducing an electrically conductive element between the electrodes. This configuration leads to ISC-induced currents across the separator, causing charge losses and localized heating within the cell. As the ISC worsens and the cell's temperature exceeds a pre-set threshold, the simulation is continued in GT-AutoLion-1D to capture the subsequent evolution of the TR process. Based on the simulation data, we perform fine-tuning of the BattBee model's ISC- and TR-related parameters, including $R_{\mathrm{ISC},1}$, $R_{\mathrm{ISC},2}$, $h_{\mathrm{ec}}$, and $\alpha_i$ for $i=1,\ldots, 4$. 

\begin{figure} [t]
 \begin{center} 
 \includegraphics[width=0.46\textwidth]{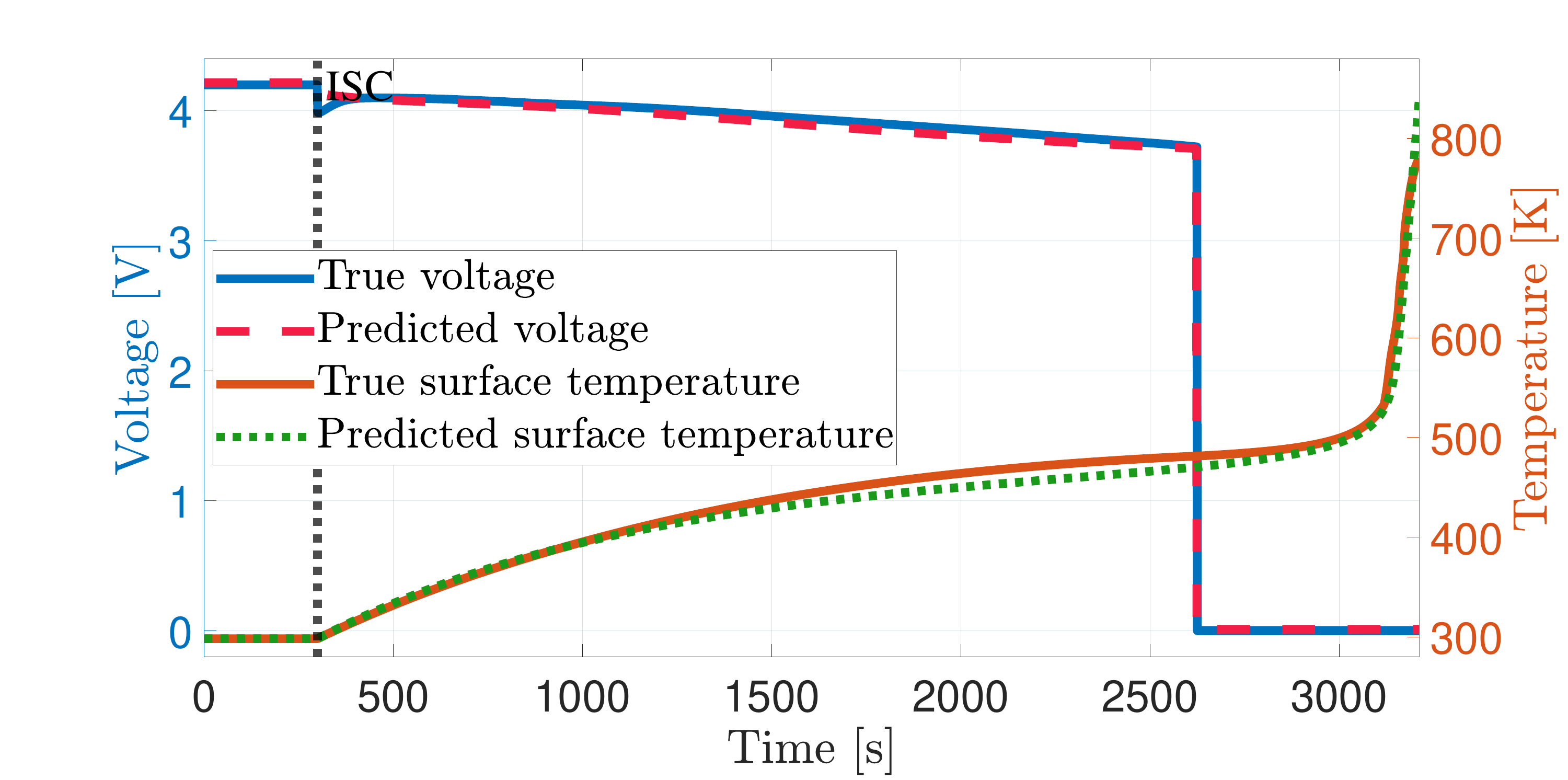}\\
 \subcaption[]{BattBee’s voltage and temperature prediction versus truth.}
 \label{fig:Sim_Val_ISC_V_T} 
 \includegraphics[width=0.46\textwidth]{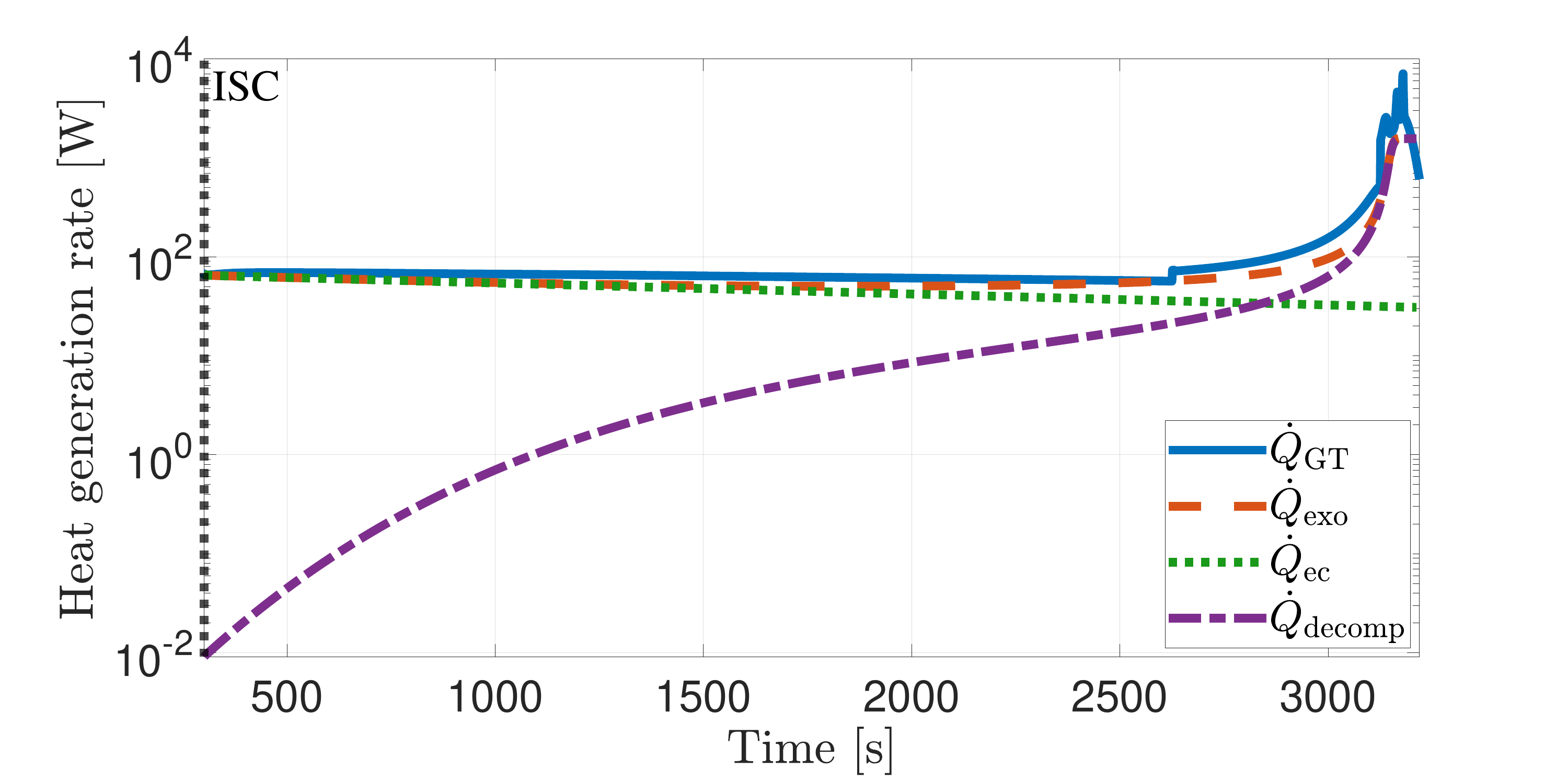}\\
 \subcaption[]{BattBee’s heat generation rates in TR.}
 \label{fig:Sim_Val_ISC_Q_dot}
 \end{center}
\caption{Simulation-based validation of the BattBee model in an ISC-induced TR event.}
\label{fig:Sim_Val_ISC}
\end{figure}

{Fig.}~\ref{fig:Sim_Val_ISC} presents the simulation and comparison. The cell, initially fully charged, experiences an ISC beginning at the $300$-th second. This triggers a voltage dip followed by a slight recovery. Subsequently, the voltage continues to decline gradually. Meanwhile, the surface temperature begins to increase at a fast pace. At the $2623$-th second, the ISC intensifies to a critical point, leading to an abrupt voltage collapse. Beyond this point, the temperature increases exponentially, driven by the onset of exothermic decomposition reactions that generate substantial heat. As shown in Fig.~\ref{fig:Sim_Val_ISC_V_T}, the BattBee model accurately predicts both voltage and temperature responses. The predictions not only closely follow the overall trends, but also capture all critical events with precision, including the ISC occurrence and the onset of the TR. Further, Fig.~\ref{fig:Sim_Val_ISC_Q_dot} illustrates the BattBee model's prediction of heat generation rates underlying the TR formation. After the ISC starts, the ISC-induced heat generation term, $\dot Q_{\mathrm{ec}}$, contributes significantly to the cell’s self-heating, while $\dot Q_{\mathrm{decomp}}$ begins to increase rapidly from a negligible level. Following the onset of the TR, $\dot Q_{\mathrm{decomp}}$ becomes the dominant contributor to heat generation. The total exothermic heat, $\dot Q_{\mathrm{exo}}$, summing up both $\dot Q_{\mathrm{ec}}$ and $\dot Q_{\mathrm{decomp}}$, closely matches the overall heat generation in the simulation by GT-AutoLion.

\subsection{Experimental Validation}

To further validate the BattBee model, we utilize open-source experimental datasets shared by Sandia National Laboratories~\cite{Lin:JES:2023}. These datasets were collected by indentation testing to induce ISCs, followed by TR events, for various cell types. They include measurements of cell voltage and temperature, among others.  The validation was performed using datasets collected from two cells, a $\chem{LiNi_{0.8}Mn_{0.1}Co_{0.1}O_2}$ cell with a capacity of $10$ Ah and dimensions of $60\times 162 \times 11$ mm, and a $\chem{LiCoO_2}$ cell with a capacity of $6.4$ Ah and dimensions of $49 \times 169 \times 8$ mm. Since the validation results for both cells appear qualitatively similar, we focus on presenting those for the $\chem{LiNi_{0.8}Mn_{0.1}Co_{0.1}O_2}$ cell.

\begin{figure} [t]
 \begin{center}
 \includegraphics[width=0.25\textwidth]{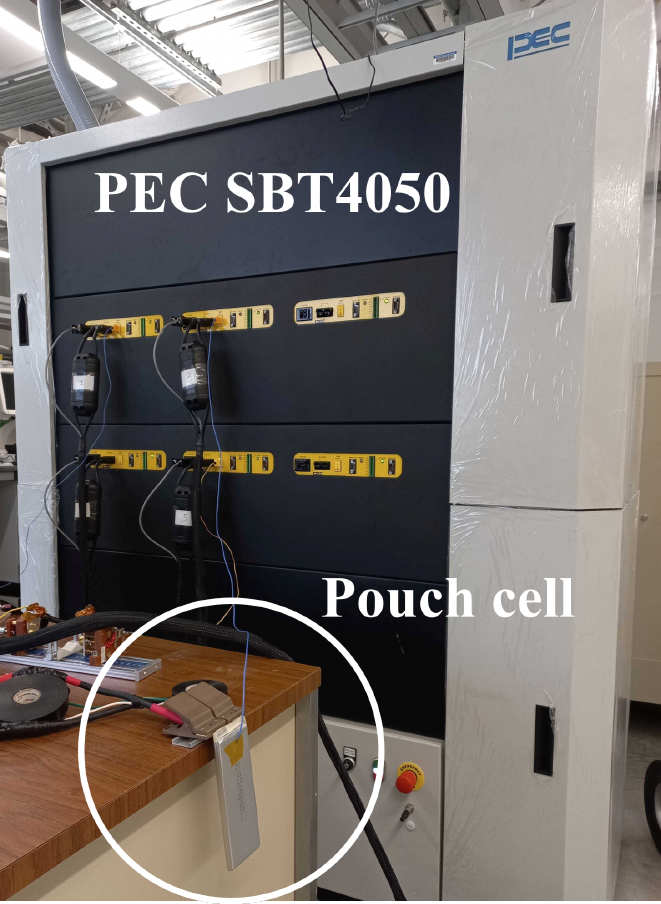} \caption{PEC\textsuperscript{\textregistered} SBT4050 battery testing system.} \label{Fig:PEC_SBT4050}
 \end{center}
\end{figure}

\begin{table*}[ht!]
\centering
\caption{Parameter values of the BattBee model identified from experiments.}
\begin{tabularx}{\linewidth}{X X X X X X X X}
\toprule
$C_b$ [F] & $C_s$ [F] & $R_b$ [$\Omega$] & $R_o$ [$\Omega$] & $C_\mathrm{core}$ [J/K] & $C_\mathrm{surf}$ [J/K] & $R_\mathrm{core}$ [K/W] & $R_\mathrm{surf,0}$  [K/W] \\ 
\midrule
13991.751 & 20003.407 & 4.721$\times10^{-3}$ & 4.726$\times10^{-3}$ & 85.539 & 10.519 & 0.834 & 9.936 \\
\bottomrule
\end{tabularx}
\label{table:BattBee_Model_Parameter_exp_cell_1}
\end{table*} 
As a first step, we acquired the same cell and conducted a series of charge/discharge tests at the University of Kansas. Performed on a PEC\textsuperscript{\textregistered} SBT4050 battery testing system with accuracy of $\pm 0.03 \%$ full scale deflection (see Fig.~\ref{Fig:PEC_SBT4050}), these tests included constant-current constant-voltage (CC/CV) charging and discharging by UDDS, LA92, SC09, and pulses. With the obtained data, we identify the BattBee model for the cell under normal fault-free operating conditions using the approach proposed in~\cite{Tu:ACC:2024}. Table~\ref{table:BattBee_Model_Parameter_exp_cell_1} summarizes the parameter estimation results. 

\begin{figure}[t]
 \begin{center}
 \includegraphics[width=0.46\textwidth]{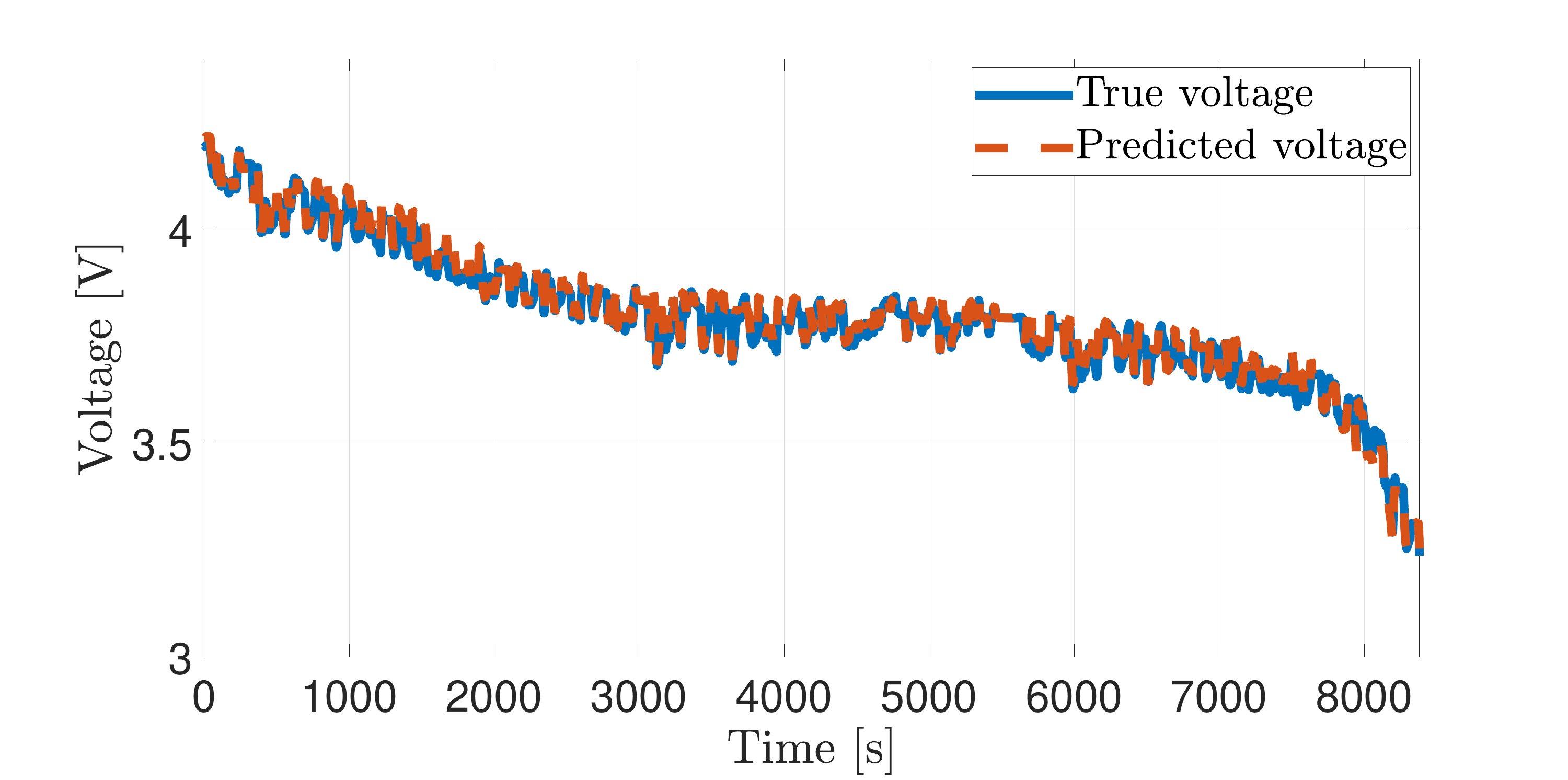}\\
 \subcaption[]{BattBee’s voltage prediction versus measurement under normal conditions.}
 \label{fig:Exp_Val_Nom_V} \includegraphics[width=0.46\textwidth]{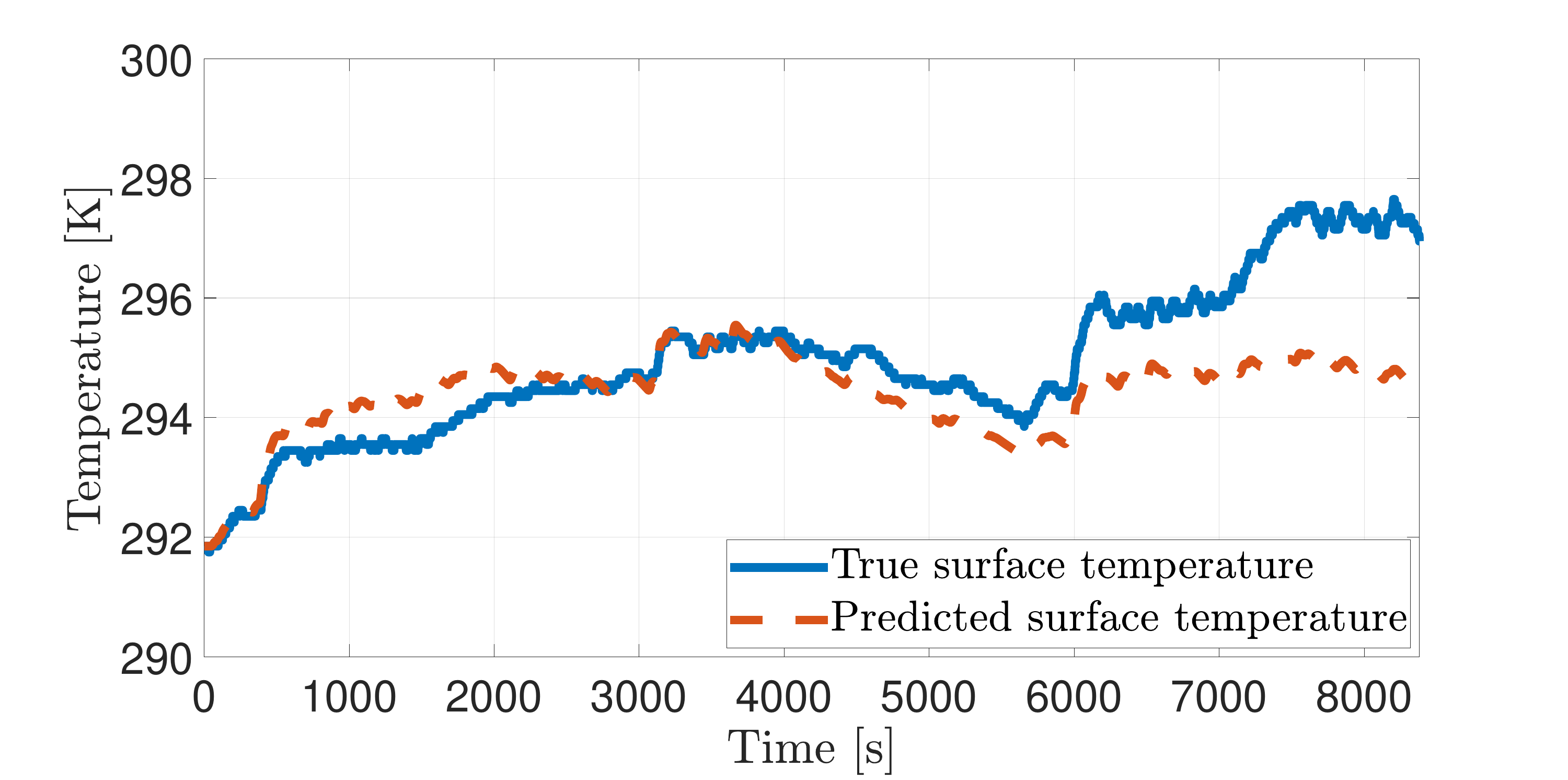}\\
 \subcaption[]{BattBee’s temperature prediction versus measurement under normal conditions.}
 \label{fig:Exp_Val_Nom_T}
 \end{center}
\caption{Experimental validation of the BattBee model under normal operating conditions.}
\label{fig:Exp_Val_Nom}
\end{figure}
Further, Fig.~\ref{fig:Exp_Val_Nom} compares the model predictions with the experimental data. Overall, the BattBee model demonstrates satisfactory accuracy in predicting both voltage and temperature. The minor discrepancy observed in the temperature prediction is attributed to the use of a constant $R_o$, which is acceptable given the model's main purpose of detecting ISC and TR events.

\begin{figure}[t]
 \begin{center} \includegraphics[width=0.4\textwidth]{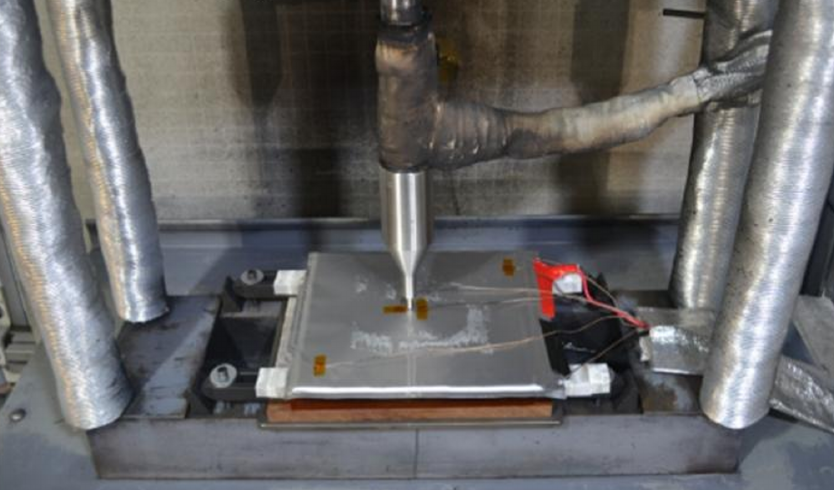}\\
 \caption{Hydraulically-driven mechanical indentation system at Sandia National Laboratories~\cite{Lin:JES:2023}.} \label{fig:Sandia_National_Lab_System}
 \end{center}
\end{figure}
We next evaluate the BattBee model's predictive capability under ISC and TR conditions. Fig.~\ref{fig:Sandia_National_Lab_System} shows the mechanically actuated, hydraulically driven indentation system at Sandia National Laboratories. This system drives an indenter to penetrate cells and induce ISCs. 
\begin{figure}[!t]
 \begin{center} \includegraphics[width=0.46\textwidth]{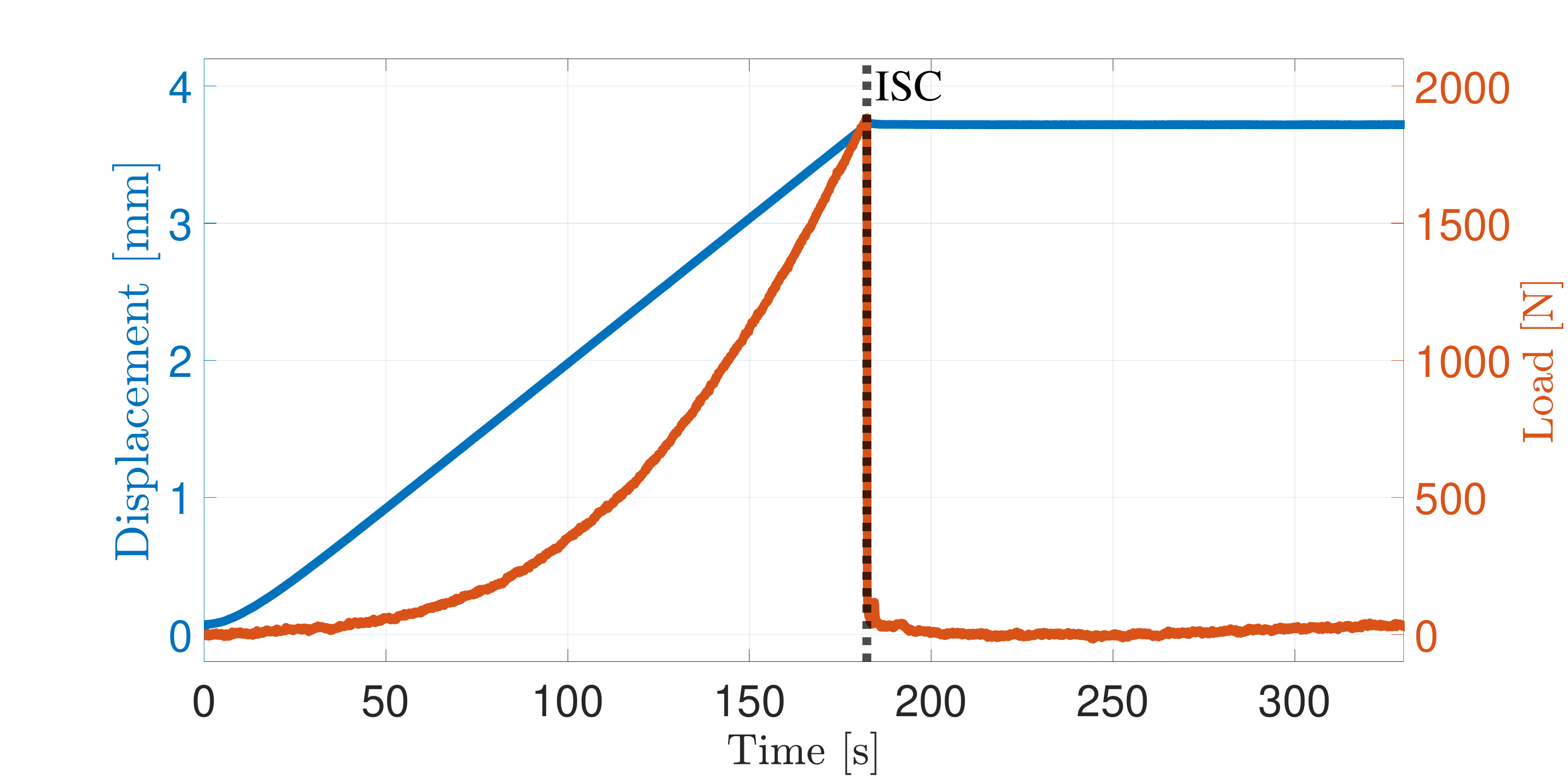}\\
 \caption{Indenter's displacement, and compressive force applied to the cell.} \label{fig:Exp_Val_ISC_Load_Disp}
 \end{center}
\end{figure}
Fig.~\ref{fig:Exp_Val_ISC_Load_Disp} presents displacement of the indenter and the corresponding compressive force applied to the cell. Upon penetration, an ISC is initiated and progressively intensifies.

\begin{figure}[t]\begin{center}\includegraphics[width=0.46\textwidth]{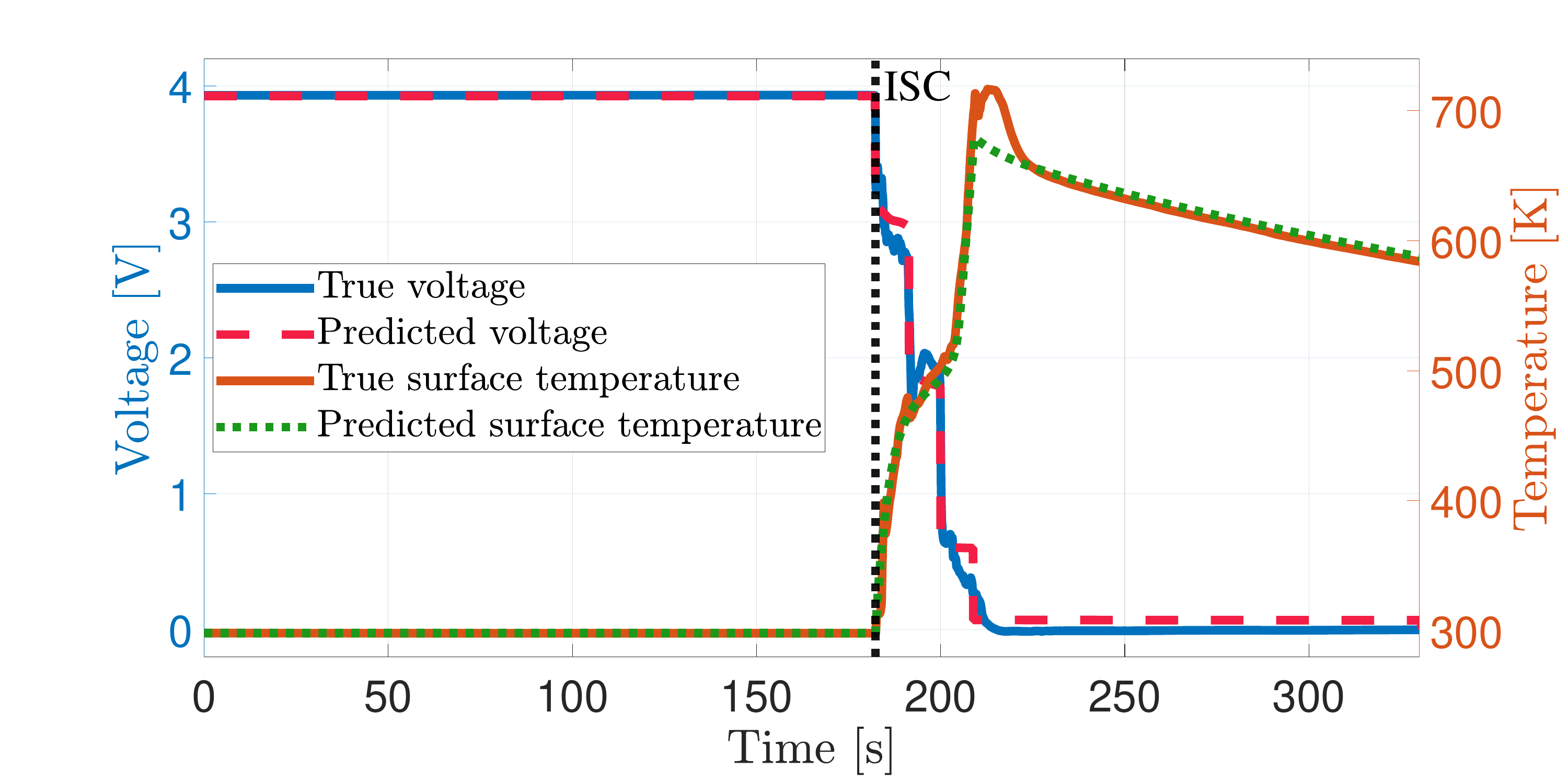}\\
 \subcaption[]{BattBee’s voltage and temperature prediction versus truth.}
 \label{fig:Exp_Val_ISC_V_T}
 \includegraphics[width=0.46\textwidth]{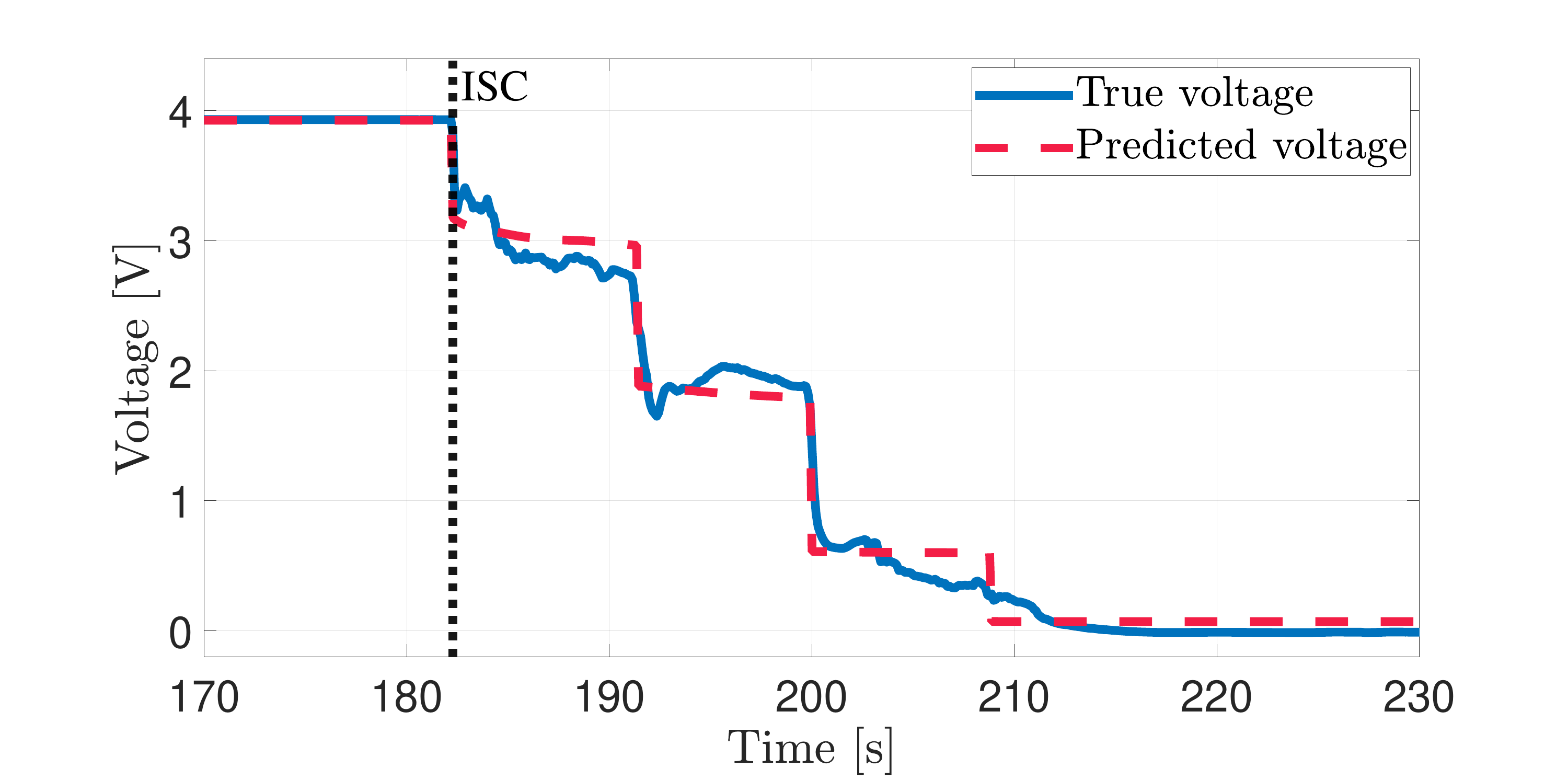}\\
 \subcaption[]{Zoomed-in view of BattBee’s voltage versus truth.}
 \label{fig:Exp_Val_ISC_V_zoom}
 \includegraphics[width=0.46\textwidth]{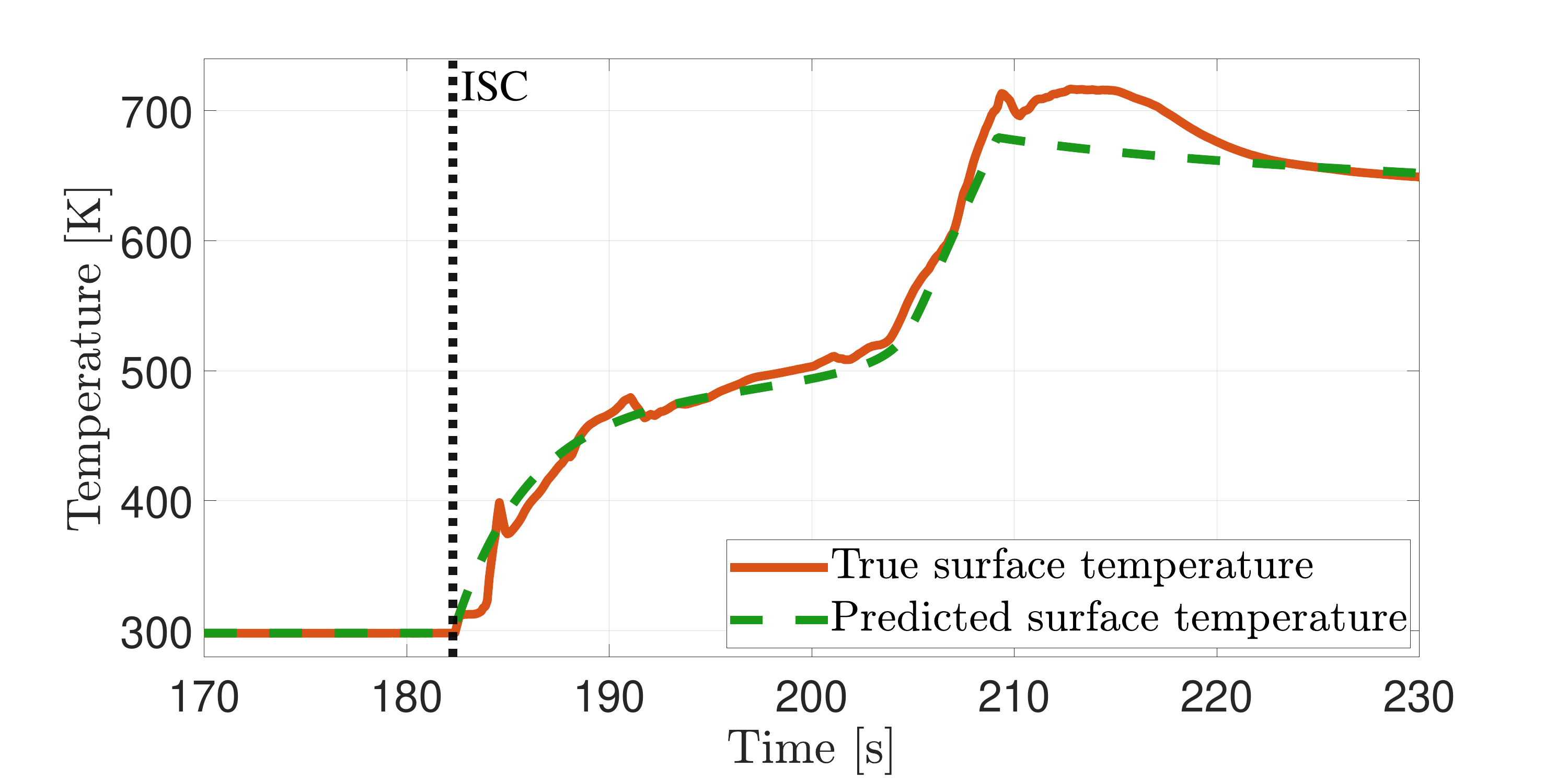}\\
 \subcaption[]{Zoomed-in view of BattBee’s temperature prediction versus truth.}
 \label{fig:Exp_Val_ISC_T_zoom}
 \end{center}
\caption{Experimental validation of the BattBee model under ISC-induced TR, with BattBee’s prediction versus measurement in voltage and temperature.}
\label{fig:Exp_Val_ISC}
\end{figure}
The resulting terminal voltage and surface temperature profiles are illustrated in Fig.~\ref{fig:Exp_Val_ISC}. As observed, the voltage drops rapidly, exhibiting distinct phases due to the growing severity of the ISC. Meanwhile, the temperature rises sharply. With fine-tuned parameters $R_{\mathrm{ISC},1}$ and $R_{\mathrm{ISC},2}$, the BattBee model accurately captures both the voltage and temperature behaviors throughout the event, as shown in Fig.~\ref{fig:Exp_Val_ISC}.

{\color{blue}To sum up, the validation results demonstrate the effectiveness of the BattBee model—it achieves sufficient accuracy with strong interpretability. Also, the model is fast for  computation and suitable for real-time safety monitoring and control, as consistently observed across all the validation cases, due to the model's simplified and low-order structure. Next up is the utilization of this model for TR and ISC detection design.}

\section{Model-Based TR and ISC Fault Detection} \label{Sec:TR-ISC-Fault-Detection}

Based on the BattBee model, we now develop a TR and ISC detection framework in this section. The proposed framework entails two key components: the design of a detection observer for residual signal generation, and the evaluation of residuals for fault detection. We perform the development for the BattBee model in its original form, and then specialize the results to a linearized case so as to obtain closed-form residual thresholds. We first carry out the development using the original nonlinear BattBee model and then specialize the results to a linearized version in order to derive closed-form conditions for fault detection.

To start with, we rewrite the BattBee model into a compact fault-explicit form as below:
\begin{align}\label{Fault-explicit-state-space-BattBee}
\left\{
\begin{aligned}
\dot  x &= (A + A_f f_1 ) x + B u + B_f f_2,\\
y &=h (x ) + Du +D_f f_3,
\end{aligned}
\right.
\end{align}
where all the matrices are defined in {{\color{blue}}Eq.} ($\star$), and
\begin{align*}
x &=\left[\begin{matrix}
V_b &V_s &T_\mathrm{core} &T_\mathrm{surf} 
\end{matrix}\right]^\top, \; 
u =\left[\begin{matrix}I & T_\mathrm{amb} &I^2 \end{matrix}\right]^\top, \\
y &=\left[\begin{matrix}V & T_\mathrm{surf} \end{matrix}\right]^\top, \; h(x) = \left[ \begin{matrix} U(\mathbf{e}_2 x) & \mathbf{e}_4 x \end{matrix} \right]^\top, \;  
f_1 = {1 \over R_{\mathrm{ISC},1} }, \\ 
f_2 &= \dot Q_\mathrm{exo}, \; 
f_3 = -\frac{R_o}{R_o + R_{\mathrm{ISC},2}} \left[ U(\mathbf{e}_2 x) + R_o I \right]. 
\end{align*}
Here, $f_i$ for $i=1,2,3$ are fault signals associated with the ISC and TR. They are zero or close to zero under normal operating conditions and become distinctly nonzero in the presence of a fault. Finally, $\mathbf{e}_i$ denotes the unit vector with a value of $1$ in the 
$i$-th entry and $0$ in all other entries.

\begin{figure*}[!htbp]  
\normalsize 
\begin{align} \tag{$\star$}
\begin{aligned}
A &=\left[\begin{matrix} -\frac{1}{R_bC_b}&\frac{1}{R_bC_b}&0&0\\\frac{1}{R_bC_s}&-\frac{1}{R_bC_s}&0&0\\0&0&-\frac{1}{R_\mathrm{core}C_\mathrm{core}}&\frac{1}{R_\mathrm{core}C_\mathrm{core}}\\0&0&\frac{1}{R_\mathrm{core}C_\mathrm{surf}}&-\frac{1}{R_\mathrm{core}C_\mathrm{surf}}-\frac{1}{R_\mathrm{surf}C_\mathrm{surf}}\end{matrix}\right] , \ \ \
A_f = \left[\begin{matrix} 0 & 0 & 0 & 0\cr
0 & -{1 \over C_s} & 0 & 0 \cr
 0 & 0 & 0 & 0\cr
 0 & 0 & 0 & 0
 \end{matrix}\right], \\ 
B&=\left[\begin{matrix}
0&0&0\\\frac{1}{C_s}&0&0\\0&0&\frac{R_o}{C_\mathrm{core}}\\0&\frac{1}{R_\mathrm{surf}C_\mathrm{surf}}&0
\end{matrix}\right], \ \ \
B_f = \left[\begin{matrix} 0 \cr 0 \cr {1 \over C_\mathrm{core}} \cr 0
 \end{matrix}\right], \ \ \
D = \left[\begin{matrix} R_o & 0 & 0 \cr 
0 & 0 & 0
 \end{matrix}\right], \ \ \
D_f = \left[\begin{matrix} 1 \cr 0
 \end{matrix}\right]. 
\end{aligned}
\end{align} 
\hrulefill 
\end{figure*}

Following~\cite{Seliger:Springer:2000,Yang:SCL:2015}, we consider the following fault detection observer for Eq.~\eqref{Fault-explicit-state-space-BattBee}:
\begin{subequations}\label{residual-dynamics}
\begin{align}
\dot {\hat x} &=A \hat x +Bu + L r,\\
r &= y- h(\hat x) - D u , 
\end{align}
\end{subequations}
where $\hat x$ is the observer's state, and $r$ is the residual in prediction of $y$. When there is no fault, $\hat x$ is an estimate of $x$, and $r$ is small in general. Defining $\tilde x = x - \hat x$ and using Eqs.~\eqref{Fault-explicit-state-space-BattBee}-\eqref{residual-dynamics}, we have
\begin{subequations}\label{residual-dynamics-rewrite}
\begin{align}
\dot {\tilde x } &= A \tilde x - L \Psi(x, \hat x)   \tilde x, \\ \label{r-xtilde-relation}
r & = \Psi(x, \hat x)  \tilde x,
\end{align}
\end{subequations}
when $f_i=0$ for $i=1,2,3$, where $L$ is a gain matrix, and 
\begin{align*}
\Psi(x, \hat x) &= \begin{bmatrix}
0 & \psi(x, \hat x)  & 0 & 0 \cr
0 & 0 & 0 & 1
\end{bmatrix}, \\ 
\psi(x, \hat x) &= { U( \mathbf{e}_2 x) - U \left( \mathbf{e}_2 \hat x \right) \over \mathbf{e}_2 \left( x - \hat x \right) }.
\end{align*}
For a LiB cell, the OCV function $U(\cdot)$ is inherently Lipschitz-continuous and monotonically increasing. It hence holds that
\begin{align*}
 \underline{\psi} \leq \psi(x, \hat x) \leq \bar \psi,
\end{align*}
where $ \underline{\psi}, \bar \psi>0$ are the minimum and maximum slopes of $U(\cdot)$. 

\begin{thm}\label{Detection-Observer-Stability}
When $f_i=0$ for $i=1,2,3$, the fault detection observer in Eq.~\eqref{residual-dynamics} is asymptotically stable if there exist $P>0$, $Q>0$, and $L$ such that 
\begin{align}\label{LMI-condition}
A^\top P +P A - H^\top L^\top P - P L H +Q \leq 0,
\end{align}
for 
\begin{align*} \color{blue}
H = 
\begin{bmatrix}
0 & \underline \psi  & 0 & 0 \cr
0 & 0 & 0 & 1
\end{bmatrix} \ \mbox{and} \ 
H = 
\begin{bmatrix}
0 & \bar \psi  & 0 & 0 \cr
0 & 0 & 0 & 1
\end{bmatrix}. 
\end{align*}
\end{thm}

\begin{proof}
Consider the Lyapunov function candidate $V(\tilde x) = {1 \over 2} \tilde x^\top P \tilde x$ (the symbol $V$ is slightly abused here to follow the notational convention in Lyapunov stability analysis). Given {Eq.}~\eqref{residual-dynamics-rewrite}, $\dot V$ is
\begin{align*}
\dot V =\tilde x^\top \left[ A - L \Psi(x, \hat x)  \right]^\top P \tilde x + \tilde x^\top P \left[ A   - L \Psi(x, \hat x)  \right] \tilde x.
\end{align*} 
According to~\cite{Zemouche:AUTO:2013}, if Eq.~\eqref{LMI-condition} holds, we have
\begin{align*}
\dot V \leq - \tilde x^\top Q \tilde x<0,
\end{align*}
for $\tilde x \neq 0$. Therefore, the observer in Eq.~\eqref{residual-dynamics} is asymptotically stable. 
\end{proof}

Theorem~\ref{Detection-Observer-Stability} indicates that, when there is no fault, $ r \rightarrow 0$ when $t\rightarrow\infty$. But it is often inadequate to compare $r$ to zero in practice to identify the occurrence of a fault, due to the transient behavior of $r$. To overcome this issue, we instead consider the following residual evaluation functions:
\begin{subequations}\label{J2-Jinfty-def}
\begin{align}
J_2 (t) &= \sqrt{ \int_0^t \|r(\tau )\|^2 d \tau }, \\ 
J_\infty (t) &=  \sup \left\{ \|r(\tau ) \|, 0\leq \tau<t \right\}, 
\end{align}
\end{subequations} 
where $ \| \cdot \|$ denotes the Euclidean norm of a vector. The design of $J_2 $ and $J_\infty$ is inspired by norms for signals---note that $J_2(\infty)$ and $J_\infty(\infty)$ are the $\mathcal{L}_2$- and $\mathcal{L}_\infty$-norms of the residual signal $r$ under fault-free conditions, respectively. %A reader is referred to~\cite{Toscano:Springer:2013} for the formal definitions of signal norms. 
The following lemma shows the upper bounds for $J_2$ and $J_\infty$. 

\begin{lem}\label{Threshold-lemma}
Suppose that $f_i=0$ for $i=1,2,3$ for Eq.~\eqref{Fault-explicit-state-space-BattBee}, and that Theorem~\ref{Detection-Observer-Stability} holds. If $\| \tilde x (0) \| \leq \delta $ and ${\epsilon \over 2} P - Q \leq 0$ for some $\epsilon >0$, then
\begin{align}\label{Detection-conditions}
J_2(t) \leq J_{2,\threshold}, \ \ \ 
J_\infty(t) \leq J_{\infty,\threshold},
\end{align}
where
\begin{align*}
 J_{2, \threshold} &= \sqrt{ \max\left\{ \bar \psi^2, 1 \right\} \cdot \lambda_{\max}(P)  \over \epsilon \cdot \lambda_{\min} (P)}  \delta , \\
J_{\infty, \threshold} &= \sqrt{ \max\left\{ \bar \psi^2, 1 \right\} \cdot \lambda_{\max} (P) \over \lambda_{\min} (P)} \delta .
\end{align*}
\end{lem}

\begin{proof}
Following Theorem~\ref{Detection-Observer-Stability} and provided that ${ \epsilon \over 2} P- Q \leq 0$, we  have
\begin{align*}
\dot V \leq - \tilde x^\top Q \tilde x \leq - \epsilon V. 
\end{align*}
By the comparison principle~\cite{Khalil:Pearson:2001}, $ V \leq e^{-\epsilon t} V (0) $. Since $V \geq \lambda_{\min} (P) \cdot \| \tilde x \|^2$, we have
\begin{align*}
 \| \tilde x \|^2 \leq  { e^{-\epsilon t} \over \lambda_{\min} (P) } V (0)  \leq { \lambda_{\max} (P) \over \lambda_{\min} (P) } \delta^2  e^{-\epsilon t}
\end{align*}
Then, by Eq.~\eqref{r-xtilde-relation}, we have
\begin{align*}
\|r\|^2 &= \tilde x^\top \Psi^\top (x, \hat x) \Psi(x, \hat x) \tilde x \leq \tilde x^\top \bar \Psi^\top \bar \Psi  \tilde x 
\\ &\leq \max\left\{ \bar \psi^2, 1 \right\} \cdot \|\tilde x\|^2 \\ 
& \leq { \max\left\{ \bar \psi^2, 1 \right\} \cdot \lambda_{\max} (P) \over \lambda_{\min} (P)} \delta^2 e^{-\epsilon t} 
\end{align*} 
With this, we can obtain $J_2(t) \leq J_2(\infty) \leq J_{2, \threshold} $ and $J_\infty(t) \leq J_\infty(\infty) \leq J_{\infty, \threshold} $, thus proving Eq.~\eqref{Detection-conditions}. 
\end{proof}

Based on Lemma~\ref{Threshold-lemma}, we introduce the following decision logic for fault detection:
\begin{equation}\label{detection-logic}
\begin{aligned}
J_p(t) &\leq J_{p,\threshold} \ \Rightarrow \ \mbox{no TR/ISC fault has occurred},\\
J_p(t) &> J_{p,\threshold} \ \Rightarrow \ \mbox{a TR/ISC fault has occurred},
\end{aligned}
\end{equation}
where $p=2$ or $\infty$. By applying this logic in conjunction with continuous residual evaluation by computing $J_2(t)$ or $J_\infty(t)$, fault occurrence can be monitored and identified. 

For the fault-free version of Eq.~\eqref{Fault-explicit-state-space-BattBee}, the only nonlinearity appears in the OCV function $U(\cdot)$. This function is amenable to piecewise linearization for LiBs~\cite{Farakhor:TTE:2023}. Leveraging this property, we can derive some direct and tighter thresholds for $J_2$ and $J_\infty$. Along this line, we approximate $U(V_s)$ as a set of $m$ linear functions:
\begin{align*}
U(V_s) = a_i V_s + b_i, 
\end{align*} 
where $a_i$ and $b_i$ are the slope and intercept of the $i$-th segment when $ \undersl{V}_{s,i} \leq V_s < \oversl V_{s,i}$ for $i=1,2,\ldots,m$. Here $ \undersl{V}_{s,i}$ and $\oversl V_{s,i}$ are the lower and upper bounds of the $i$-th segment. Then, we can translate Eq.~\eqref{Fault-explicit-state-space-BattBee} into a collection of linear models, $\mathcal{M}_i$ for $i=1,2,\ldots, m$:
\begin{align}\label{Fault-explicit-state-space-BattBee-linearized}
\mathcal{M}_i: \left\{
\begin{aligned}
\dot  x &= (A + A_f f_1 ) x + B u + B_f f_2,\\
z &= C_i x + Du +D_f f_3,
\end{aligned}
\right.
\end{align}
where
\begin{align*}
z &= \begin{bmatrix}
V - b_i \cr T_\mathrm{surf}
\end{bmatrix}, \ \ \ 
C_i = \begin{bmatrix}
0 & a_i & 0 & 0 \cr
0 & 0 & 0 & 1
\end{bmatrix}. 
\end{align*}
Given $\mathcal{M}_i$, a linear fault detection observer can be designed following Eq.~\eqref{residual-dynamics} and~\cite{Patton-Frank-Clark:Springer:2000,Li-Ding:AUTO:2017,Xu:CCA:2007}:
\begin{subequations}\label{residual-dynamics-linearized}
\begin{align}\label{State-observer-linearized}
\dot {\hat x} &= A \hat x + B u + L_i r,  \\
r &= z - C_i \hat x - Du.
\end{align}
\end{subequations}
In the case of no fault, we have 
\begin{align}\label{State-residual-when-no-fault}
\dot {\tilde x } & = \tilde A_i \tilde x, \ \ \
r  = C_i  \tilde x,
\end{align}
where $ \tilde A_i = A   - L_i C_i $. As a result, 
\begin{align*}
r = C_i e^{\tilde A_it} \tilde x(0).
\end{align*}
To evaluate the residual's behavior, we continue to use $J_2$ and $J_\infty$ in Eq.~\eqref{J2-Jinfty-def}. 

\begin{lem}\label{Threshold-lemma-linearized}
Suppose that $f_i=0$ for $i=1,2,3$ for Eq.~\eqref{Fault-explicit-state-space-BattBee-linearized}, that $(A,C_i)$ is detectable for $i=1,2,\ldots,m$, and that $\| \tilde x (0) \| \leq \delta $. Then, in the regime of $\mathcal{M}_i$, 
\begin{align}\label{Detection-conditions-linearized}
J_2(t) \leq J_{2,\threshold}^i, \ \ \ 
J_\infty(t) \leq J_{\infty,\threshold}^i,
\end{align}
where
\begin{align*}
 J_{2, \threshold}^i &= \left \| \sqrt{ W_i} \right \| \delta , \ \ W_i = \int_0^\infty  e^{\tilde A_i^\top t} C_i^\top C_i e^{\tilde A_i t} dt, \\
J_{\infty, \threshold}^i &=  \sup_\tau \left\{ \left \|   C_i e^{\tilde A_i \tau}  \right\| \delta  , \; 0 < \tau < \infty \right\} . 
\end{align*}
\end{lem}

\begin{proof}
We have
\begin{align*}
J_2(t) &\leq  \sup_{\| \tilde x (0) \| \leq \delta} \sqrt{ \int_0^\infty \|r (\tau)\|^2 d \tau } \\
&=  \sup_{\| \tilde x (0) \| \leq \delta} \sqrt{ \int_0^\infty \tilde x(0)^\top e^{\tilde A_i^\top t} C_i^\top C_i e^{\tilde A_it} \tilde x(0) dt } \\
&= \sup_{\| \tilde x (0) \| \leq \delta} \sqrt{ \tilde x(0)^\top \underbrace{ \int_0^\infty  e^{\tilde A_i^\top t} C_i^\top C_i e^{\tilde A_it} dt }_{W_i} \tilde x(0) } \\
&= \left \| \sqrt{ W_i} \right \| \delta.
\end{align*}
Similarly,
\begin{align*}
J_\infty(t)&  \leq \sup_{\| \tilde x (0) \| \leq \delta} \left\{ \|r (\tau)\|  , 0 \leq \tau < \infty \right\}\\
&= \sup_{\| \tilde x (0) \| \leq \delta} \left\{ \sqrt{ \tilde x(0)^\top  e^{\tilde A_i^\top \tau} C_i^\top C_i e^{\tilde A_i \tau} \tilde x(0) } ,  0 \leq \tau < \infty \right\} \\
&\leq \sup_\tau \left\{ \left \|  C_i e^{\tilde A_i \tau} \right \| \delta  , 0 \leq \tau < \infty \right\}. 
\end{align*}
This completes the proof.
\end{proof}

Note that the notation $\|\cdot\|$ is slightly abused above to also represent the induced 2-norm of a matrix, yet without causing confusion from the context. Furthermore, it is known that $W_i$ is the observability Gramian for Eq.~\eqref{State-residual-when-no-fault}, satisfying the following Lyapunov equation:
\begin{align*} 
\tilde A_i^\top W_i + W_i \tilde A_i = - C_i^\top C_i. 
\end{align*} 
By solving this Lyapunov equation, one can directly compute $W_i$ and then $ J_{2,\threshold}^i$. Finally, by leveraging Lemma~\ref{Threshold-lemma-linearized}, we can employ the decision logic in Eq.~\eqref{detection-logic} to perform TR and ISC detection.

We have the following remarks to address some practical considerations in implementing the above fault detection approach based on the piecewise linearized models. First, to determine $\mathcal{M}_i$ corresponding to the operating condition of the LiB cell, the current state can be continuously estimated using Eq.~\eqref{State-observer-linearized} when no fault has been identified. Second, inevitable model inaccuracy and measurement noises can mildly alter the behavior of the residual signals to cause $J_2$ and $J_\infty$ to exceed their theoretical values under normal operating conditions, potentially triggering false alarms. To alleviate this, we suggest to modify $J_2$ as
\begin{align}\label{eqn:J2-with-forgetting-factor}
J_2(t) = \sqrt{\int_0^t \eta^{(t - \tau)} \| r (\tau)\|^2 \, d\tau}, 
\end{align}
where $0<\eta<1$ is a forgetting factor to give exponentially decreasing weight to older residual values. In the discrete-time domain, $J_2$ can be recursively computed via
\begin{align*}
J_2 (k+1 ) = \sqrt{ \eta^{\Delta  t} J_2^2(k)+\| r (k+1)\|^2 \Delta t },
\end{align*}
where $k$ is the discrete time index and $\Delta t$ is the sampling period. The choice of $\eta$ should balance the accuracy and robustness against uncertainty in fault detection; typically, $\eta$ is chosen to be slightly less than 1. In general, $J_\infty$ is less sensitive to the impact of model and measurement uncertainty. However, it is still beneficial to slightly increase $J_{\infty, \threshold}^i$ to make the fault detection more robust. Finally, for the sake of conservatism, one can select $J_{p,\threshold}$ for $p=2$ or $\infty$ as 
\begin{align*}
  J_{p,\threshold} = \max_{i} \left\{ J_{p,\threshold}^i, i=1,2,\ldots, m \right\}. 
\end{align*}
This conservative choice will help reduce the chance of false alarming in practical applications.

\section{Validation of TR and ISC Detection} \label{Sec:TR-ISC-Fault-Detection-Validation}

This section validates the TR and ISC detection approach proposed in Section~\ref{Sec:TR-ISC-Fault-Detection}. We perform the validation using both simulation and experimental data, aligned with the model validation presented in Section~\ref{sec:BattBee-Model-Validation}.

\subsection{Simulation-Based Validation of TR and ISC Detection}

Our first validation leverages the BattBee model identified in Section~\ref{sec:BattBee-Model-Validation-Sim} as well as the corresponding simulation datasets. We focus on evaluating $J_2$ and choose $J_{2,\threshold}$ as in Lemma~\ref{Threshold-lemma-linearized} using the piecewise-linearized version of the BattBee model. The evaluation of $J_2$ follows Eq.~\eqref{eqn:J2-with-forgetting-factor}, employing a forgetting factor $\eta=0.95$ to enhance the robustness of the detection against model and measurement uncertainty, as highlighted in Section~\ref{Sec:TR-ISC-Fault-Detection}. In the simulation, the initial state estimation error is upper bounded by $  \begin{bmatrix} 0.01 & 0.01 & 0.1 & 0.1 \end{bmatrix}^\top$.  
The fault detection observer follows Eq.~\eqref{eqn:J2-with-forgetting-factor}, where a Kalman gain is used for its optimality and convenience. 

\begin{figure}[t]
 \begin{center}
\includegraphics[width=0.46\textwidth]{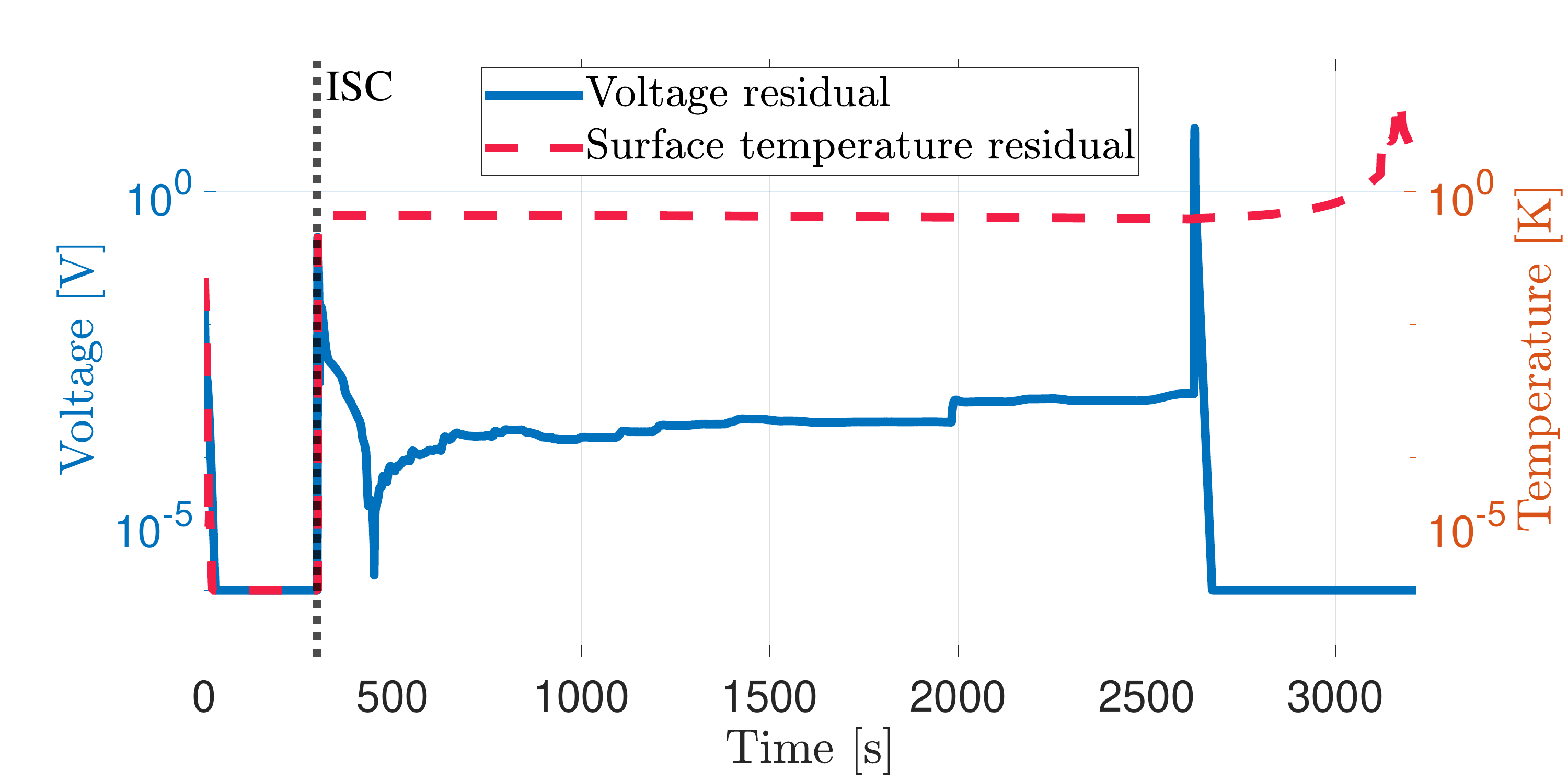}\\
 \subcaption[]{Residual signals for voltage and temperature estimates.}
 \label{fig:Sim_Val_FD_r} 
 \includegraphics[width=0.46\textwidth]{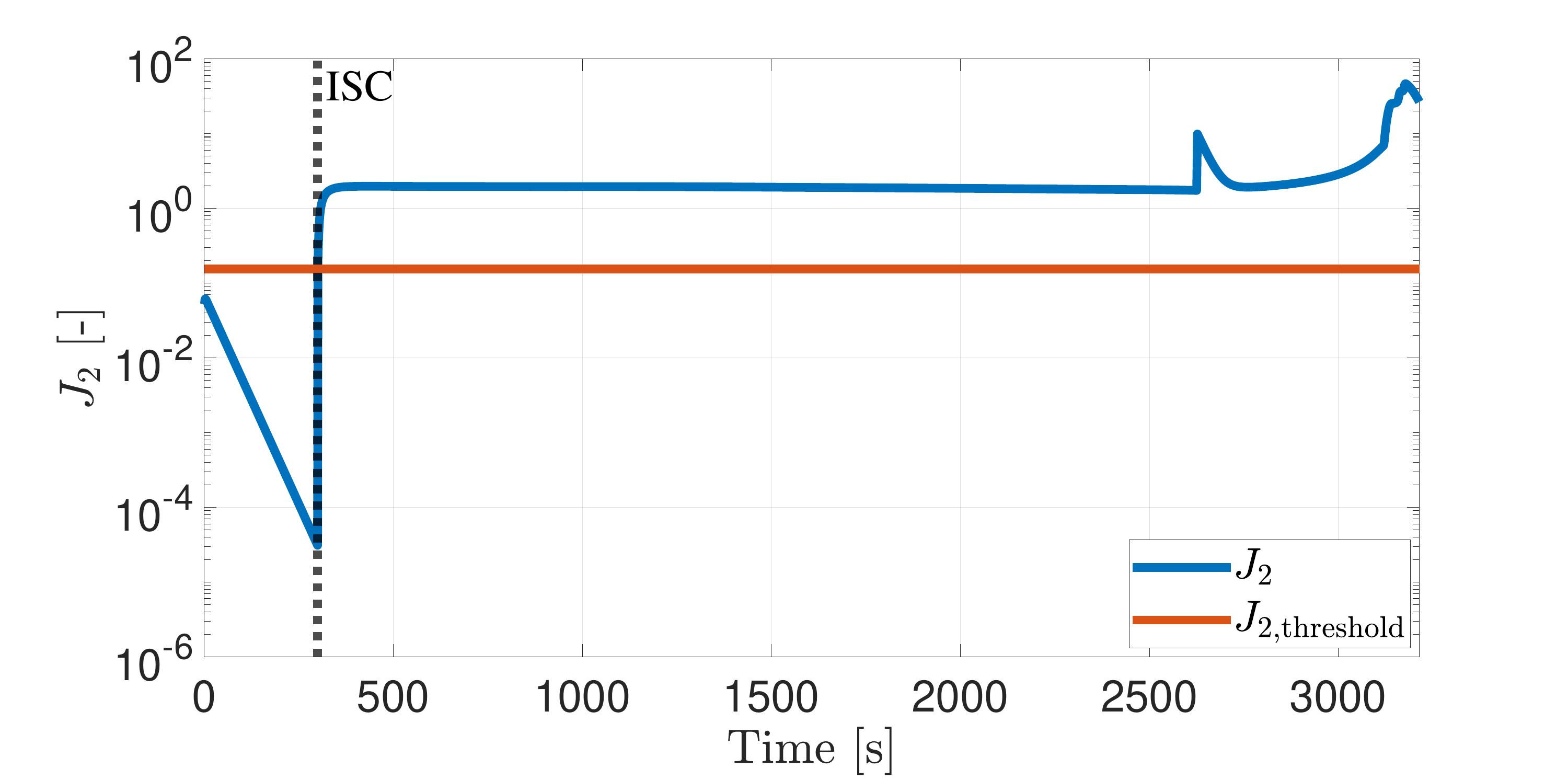}\\
 \subcaption[]{Fault detection result based on the residual evaluation function ($J_2$) and threshold.}
 \label{fig:Sim_Val_FD_J} 
 \includegraphics[width=0.46\textwidth]{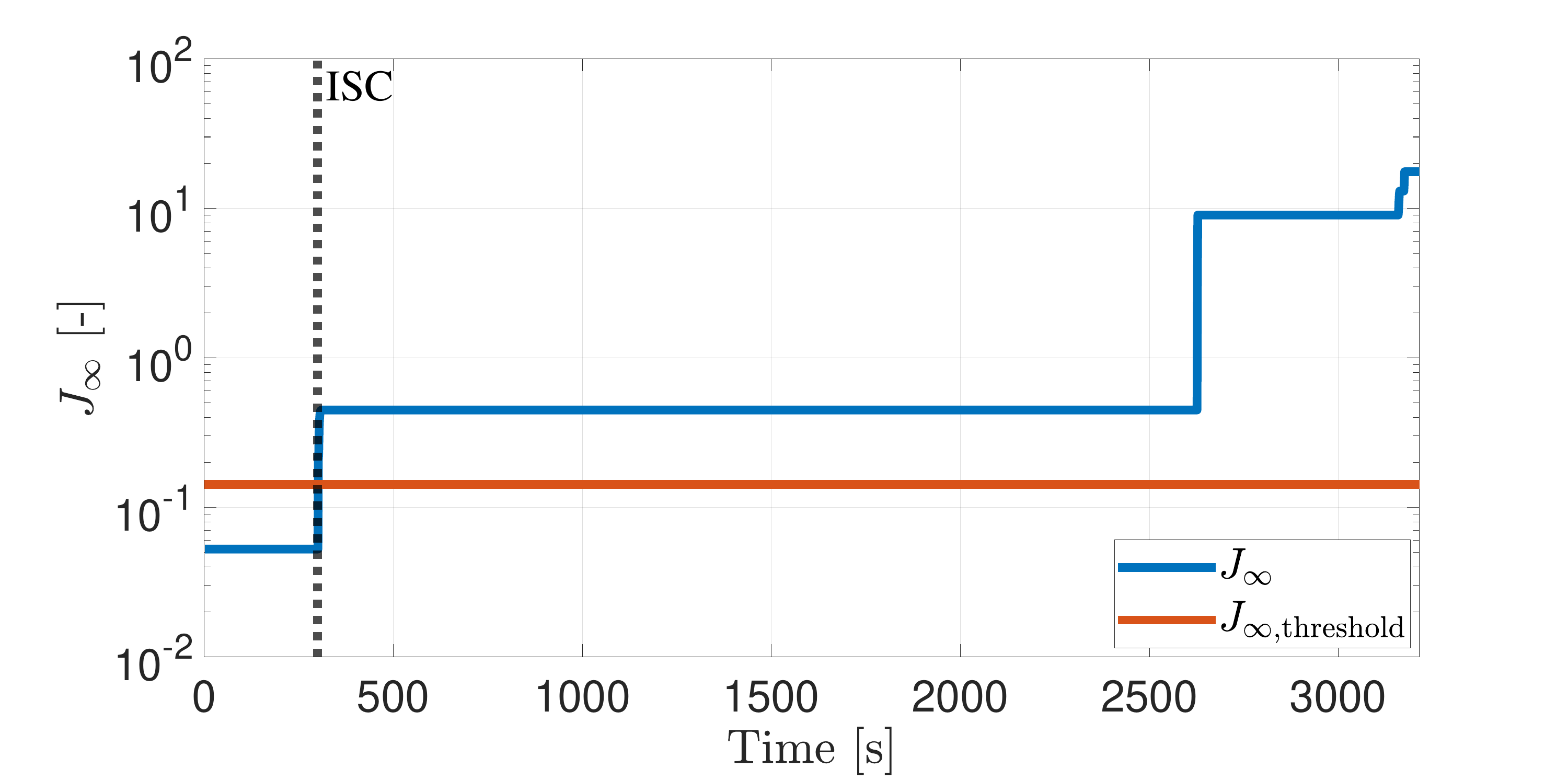}\\
 \subcaption[]{Fault detection result based on the residual evaluation function ($J_\infty$) and threshold.}
 \label{fig:Sim_Val_FD_J_inf}
 \end{center}
\caption{Simulation-based validation of the model-based fault detection system in an ISC-induced TR event.}
\label{fig:Sim_Val_FD}
\end{figure} 

 {Fig.}~\ref{fig:Sim_Val_FD} presents the validation results.
The ISC begins at the $300$-th second, causing a gradual decline in the terminal voltage and a rise in the surface temperature. The ISC becomes crucially severe at the $2623$-th second to induce exothermic reactions, with the TR unfolding completely.  
{Fig.}~\ref{fig:Sim_Val_FD_r} illustrates the residual signals for the terminal voltage and surface temperature, generated by the fault detection observer. On the initial occurrence of the ISC, the residuals for both exhibit a noticeable jump and remain higher than usual during the process. When the ISC becomes fully fledged, the residual for the terminal voltage plummets; meanwhile, the residual for the temperature shows a sharp increase, albeit with a delay of about $100$ seconds. The fault detection results are displayed in Fig.~\ref{fig:Sim_Val_FD_J} and~\ref{fig:Sim_Val_FD_J_inf}. As can be seen, $J_2$ and $J_\infty$  exceed $J_{2,\threshold}$ and $J_{\infty,\threshold}$ immediately upon the initial ISC occurrence, thereby triggering a fault alarm. Subsequently, $J > J_{\threshold}$ hold continuously, providing persistent indication of the ISC and impending TR event. When the full ISC occurs, $J_2$ and $J_\infty$ exceed $J_{2,\threshold}$ and $J_{\infty,\threshold}$ by even larger margins. We emphasize that the fault is detected well before the TR fully unfolds, implying the model’s capability for early and reliable fault detection.

\subsection{Experimental Validation of TR and ISC Detection}

For this validation, we use experimental datasets collected by Sandia National Laboratories~\cite{Lin:JES:2023}, along with the BattBee model identified in Section~\ref{BattBee-Model-Governing-Equations}. The detection approach uses the piecewise linearized version of the BattBee model and evaluates $J_2$ according to Eq.~\eqref{eqn:J2-with-forgetting-factor} with a forgetting factor of $\eta = 0.95$. 

\begin{figure}[t]
 \begin{center}
 \includegraphics[width=0.46\textwidth]{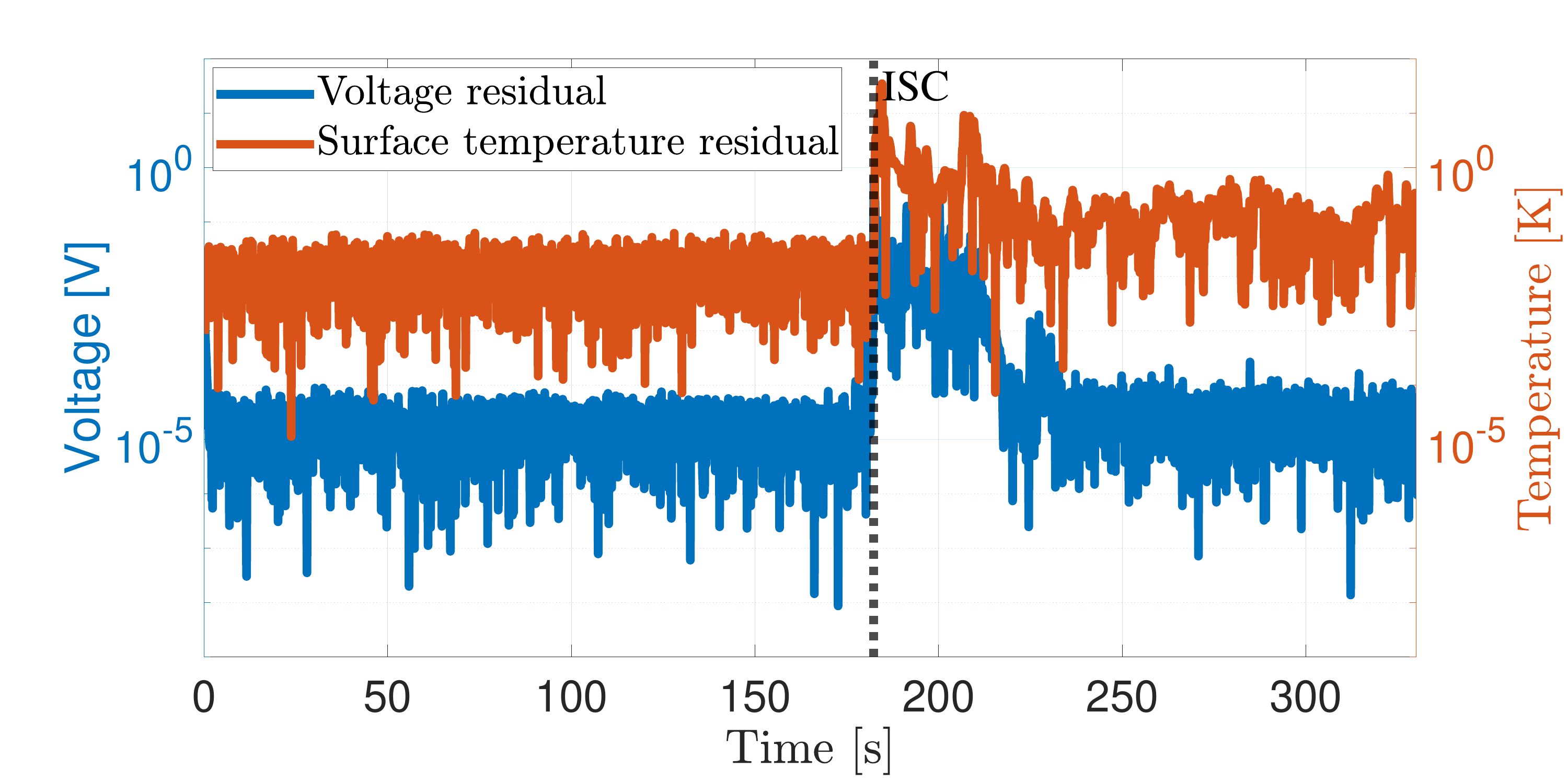}\\
 \subcaption[]{Residual signals for voltage and temperature estimates.}
 \label{fig:Exp_Val_FD_r}
 \includegraphics[width=0.46\textwidth]{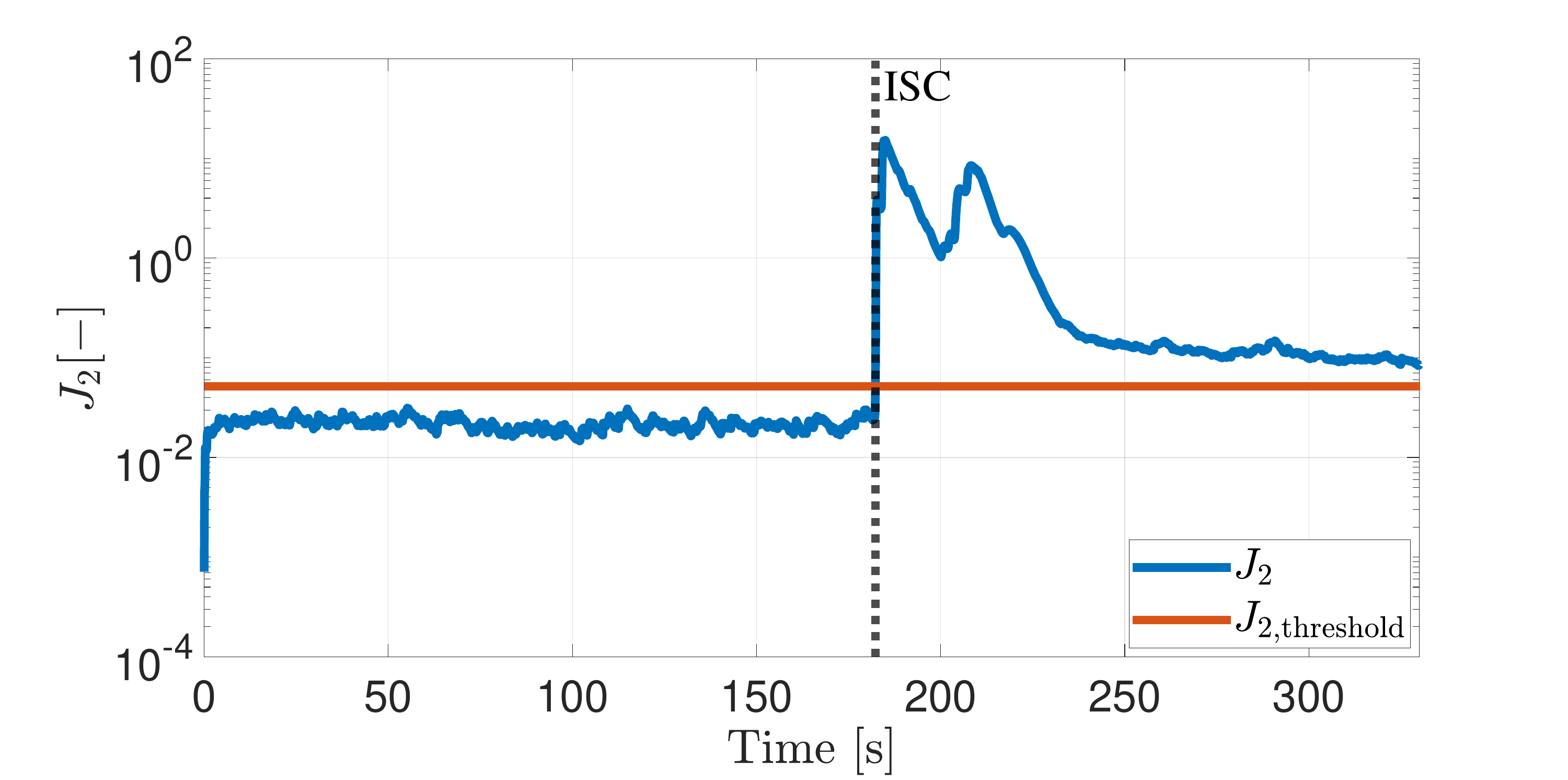}\\
 \subcaption[]{Fault detection result based on the residual evaluation function ($J_2$) and threshold.}
 \label{fig:Exp_Val_FD_J}
 \includegraphics[width=0.46\textwidth]{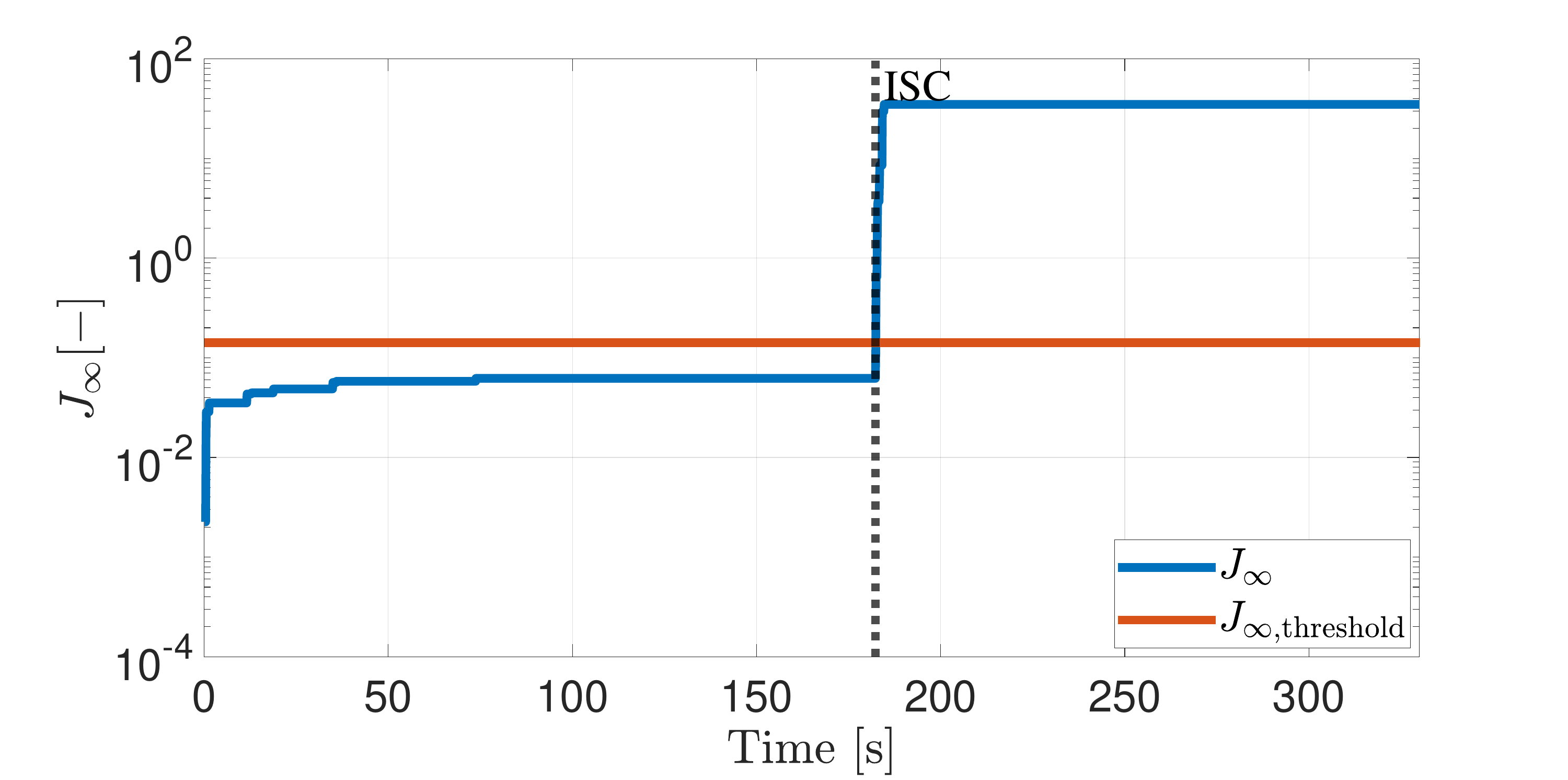}\\
 \subcaption[]{Fault detection result based on the residual evaluation function ($J_\infty$) and threshold.}
 \label{fig:Exp_Val_FD_J_inf}
 \end{center}
\caption{Experimental validation of the model-based fault detection system in an ISC-induced TR event.}
\label{fig:Exp_Val_FD}
\end{figure}
Fig.~\ref{fig:Exp_Val_FD} presents the validation results.  
The indenter induces an ISC condition in the cell at around the $182$-th second. This condition deteriorates, ultimately leading to a complete failure, as evidenced by the rapid voltage drop. Concurrently, the cell temperature rises sharply, initiating a TR event. The BattBee model demonstrates strong agreement with the experimental data, confirming its predictive accuracy. When the fault detection observer is applied to the BattBee model with the experimental data, Fig.~\ref{fig:Exp_Val_FD_r} displays the resulting residuals for terminal voltage and surface temperature. Although the residuals exhibit fluctuations due to model and measurement uncertainty, they consistently rise to higher levels following the onset of the ISC. In Fig.~\ref{fig:Exp_Val_FD_J} and~\ref{fig:Exp_Val_FD_J_inf}, $J_2$ and $J_\infty$ are evaluated continuously, which surpass $J_{2,\threshold}$ and $J_{\infty,\threshold}$ right after the ISC occurs, leading to an effective detection of both the ISC and the subsequent TR event.

\begin{figure}[t]
 \begin{center} 
 \includegraphics[width=0.46\textwidth]{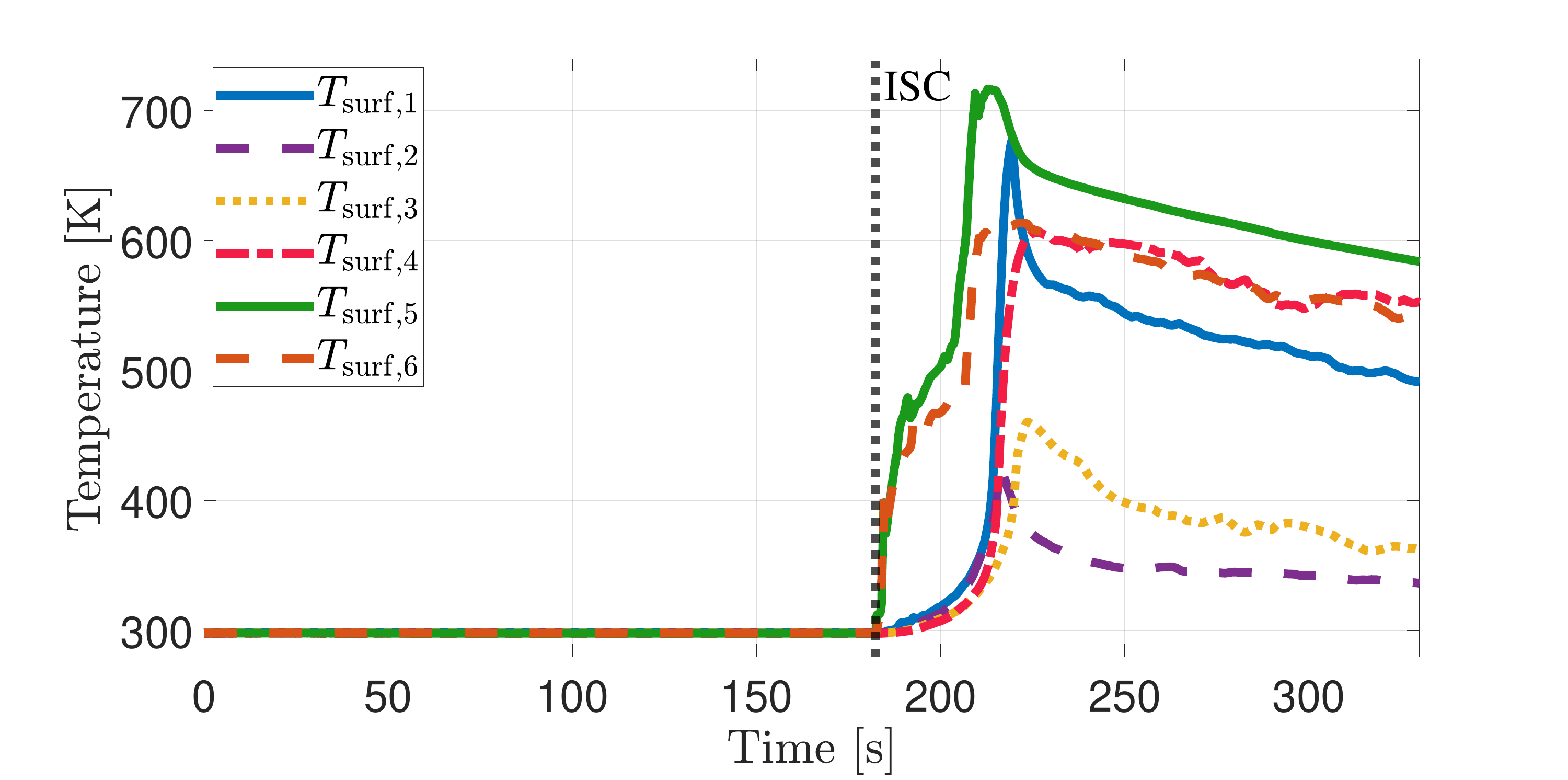}\\
 \caption{Temperature evolution at six surface locations of a pouch cell under localized ISC-induced heating.}\label{Fig:Exp_Val_FD_VaryT_T}
 \end{center}
\end{figure}

For the pouch cell under testing, the surface temperature was measured at six different locations: $T_\mathrm{surf,1}$ near the positive terminal, $T_\mathrm{surf,2}$ near the negative terminal, $T_\mathrm{surf,3}$ at the bottom, $T_\mathrm{surf,4}$ at the top, $T_\mathrm{surf,5}$ above the penetration point, and $T_\mathrm{surf,6}$ below the penetration point. {Fig.}~\ref{Fig:Exp_Val_FD_VaryT_T} illustrates the temperature evolution at these six locations. Each measurement exhibits different response speeds and magnitudes of temperature rise, reflecting the localized heating effects caused by the ISC. In general, the closer a sensor is to the penetration point, the more rapidly the temperature increases. Nonetheless, all locations experience substantial temperature escalation.

\begin{figure}[t]
 \begin{center} 
 \includegraphics[width=0.46\textwidth]{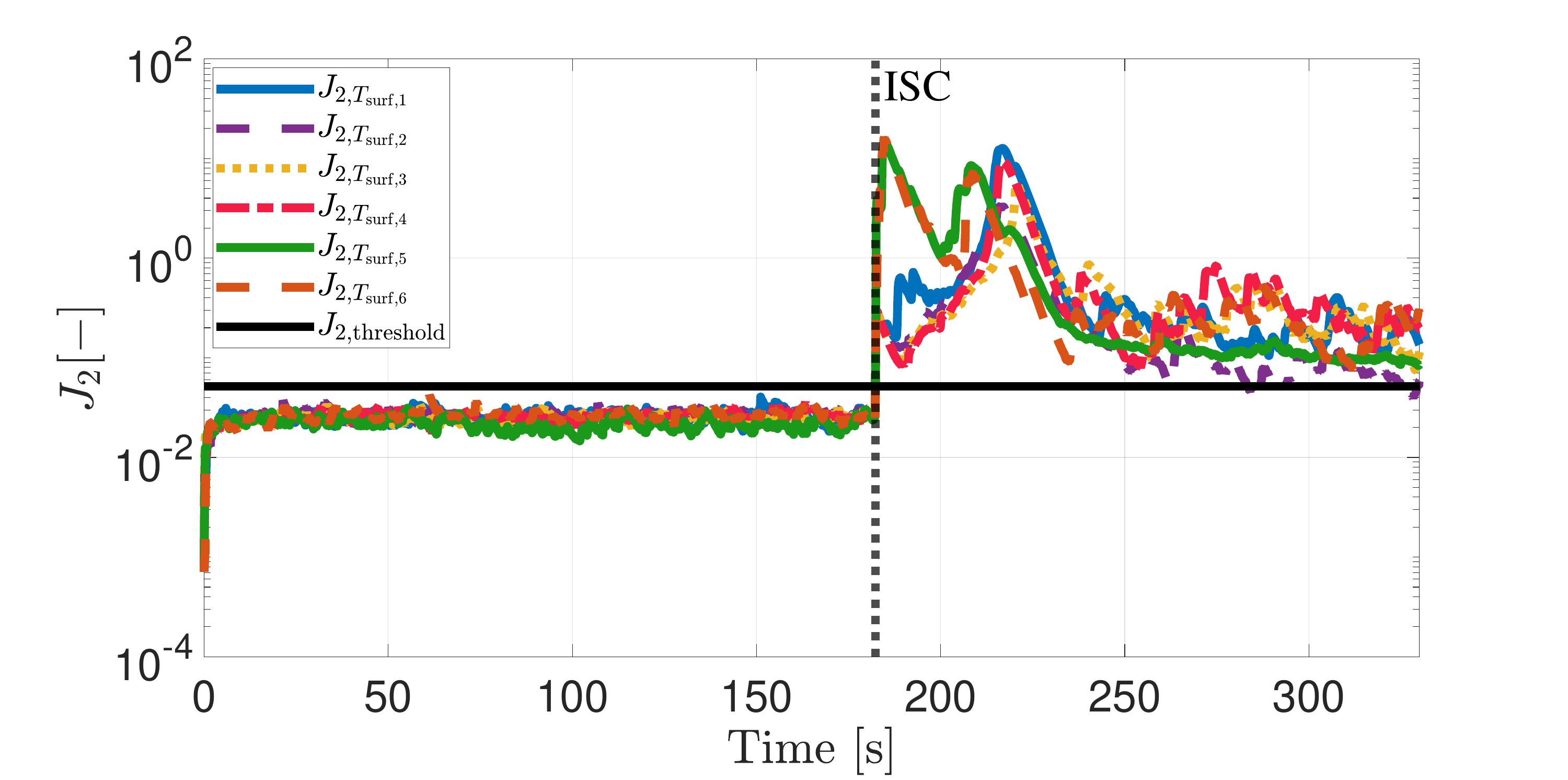}\\
 \subcaption[]{Fault detection result based on the residual evaluation function ($J_2$) and threshold.}
 \label{Fig:Exp_Val_FD_VaryT_J_2}
 \includegraphics[width=0.46\textwidth]{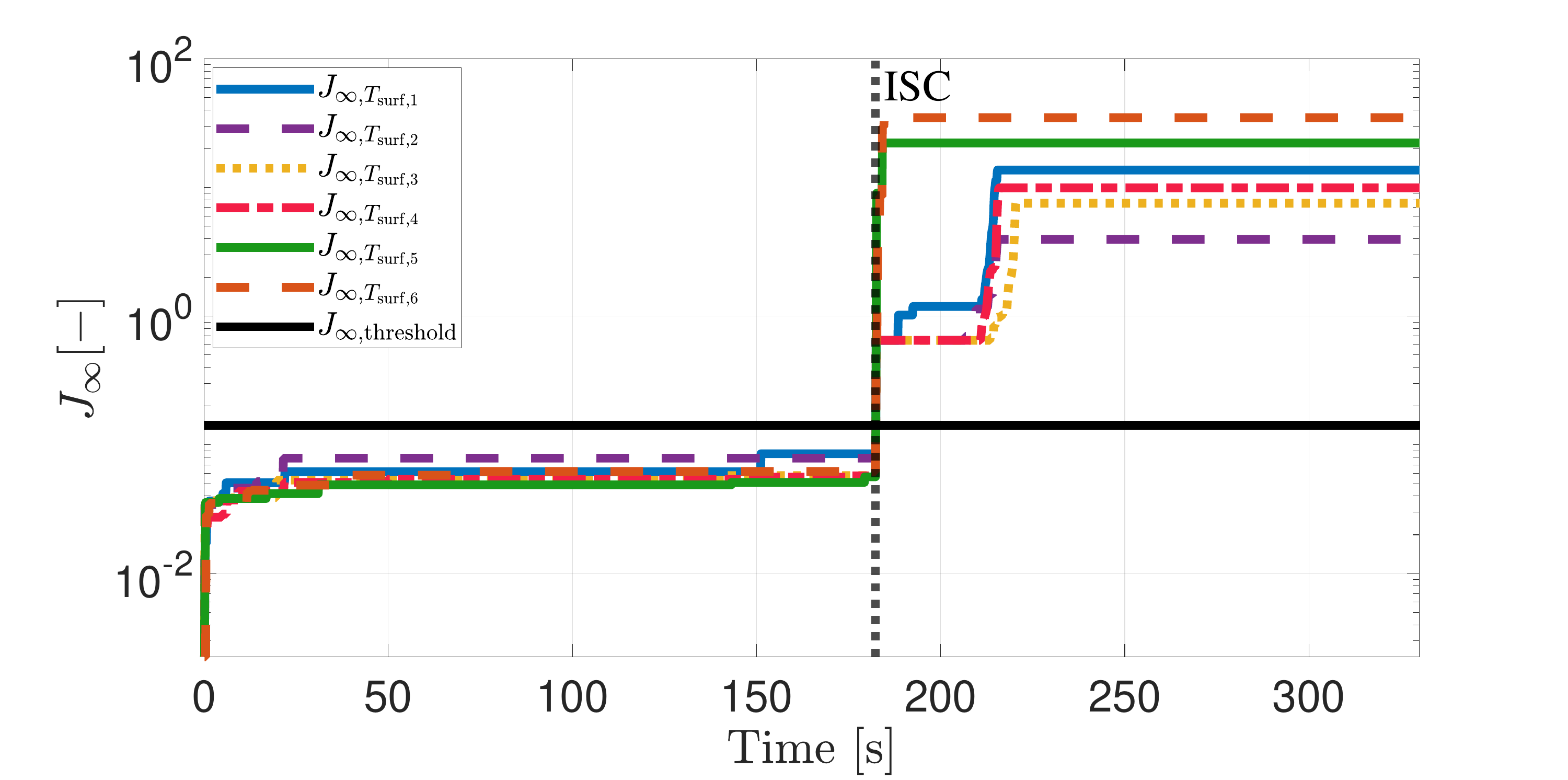}\\
 \subcaption[]{Fault detection result based on the residual evaluation function ($J_\infty$) and threshold.}
 \label{Fig:Exp_Val_FD_VaryT_J_inf}
 \end{center}
\caption{Fault detection results using temperature measurements from six surface locations.}
\label{Fig:Exp_Val_FD_VaryT_J}
\end{figure}
 {Fig.}~\ref{Fig:Exp_Val_FD_VaryT_J} presents the fault detection performance using the temperature measurements from each of these locations by evaluating the corresponding residual evaluation functions. In all cases, the ISC and subsequent TR are successfully detected. Notably, detection is achieved even with relatively milder temperature increases observed in $T_\mathrm{surf,1}$, $T_\mathrm{surf,2}$, $T_\mathrm{surf,3}$, and $T_\mathrm{surf,4}$, despite their locations farther from the penetration point. Across all locations, detection occurs more than $30$ seconds prior to the cell reaching its peak temperature. This lead time enhances safety assurance and helps mitigate the consequences of TR on users and surrounding assets. 

 In summary, the validation results demonstrate the robustness and effectiveness of the proposed detection approach. Both ISC and TR conditions are reliably identified across all cases, with lead times of tens of seconds or more before the full-scale outbreak of the event. Such early warning capability is critical for ensuring user safety in practice. Moreover, the detection approach is computationally fast, resulting from the simplicity of both the BattBee model and the observer. 

\section{Conclusion} \label{Sec:Conclusion}

Safety is at high stakes as LiBs are enabling the sweeping advancement of electric vehicles and will continue to power future industries and societies. However, LiBs have yet to be fully safe, even though they provide high energy density and long cycle life compared to other battery technologies. Critical concerns lie in the threat of TR events that are often triggered by ISCs. This challenge has prompted ever-growing research aimed at developing effective solutions. While the literature on ISC-induced TR modeling is increasing, there is still a lack of dynamic models with structural simplicity and computational efficiency necessary for practical battery management systems.

In this paper, we propose a novel equivalent circuit model called   BattBee  to fill this gap. By design, the BattBee model is a circuit analog simplifying an electrochemical model (SPM) coupled with an ISC mechanism. With this physical interpretability, the model is capable of simulating various key phenomena arising in TR initiated by ISC conditions---including rapid charge depletion, voltage collapse, heat generation and buildup, and temperature escalation. 
Building on the BattBee model, we further propose an observer-based fault detection approach to achieve real-time monitoring and detection of ISC and TR events. This approach includes explicit detection criteria and decision logics, well-suited for practical implementation. Validation using both simulation and experimental data demonstrates that the BattBee model delivers high fidelity and predictive accuracy, while the detection framework effectively enables early detection of ISC and TR conditions.

Looking ahead, there are several meaningful directions to expand this study. First, it is of interest to further enhance the descriptive capabilities of the BattBee model. For instance, the model can be modified to incorporate SoC- and temperature-dependent internal resistance. Drawing upon~\cite{Biju:AE:2023}, the BattBee model can be refined to account for more complex dynamics and capture ISC-induced TR behavior under high C-rates, with potential applications to electric aircraft. 
 Second, although the BattBee model is a 0D model capable of demonstrating overall voltage and surface temperature dynamics, it does not fully capture spatial temperature or current density distribution  
observed during an ISC or TR event. A potential approach to address this limitation is through multi-region lumped modeling. Note that the associated increase in model complexity, additional measurement requirements, and difficulty in parameter identification would also require further research. 
Third, for both modeling and fault detection, integrating physics with machine learning is likely to open up new avenues for improving accuracy and practical deployability, as suggested by some recent studies, e.g.,~\cite{Tu:AE:2023}. 
 Fourth, since the BattBee model  
neglects the influence of battery degradation, incorporating battery aging dynamics can be another useful direction to improve its descriptive capability and applicability throughout the cell's lifetime. 
Finally, future developments should be validated through experiments on LiB cells with different geometries and chemistries.

\section*{Declaration of Competing Interest}
The authors declare that they have no known competing financial interests or personal relationships that could have appeared to influence the work reported in this paper. 

\section*{Data Availability}
Data will be made available on request. 

\section*{Acknowledgement}
This work was supported in part by the U.S. Department of Energy under Award DE-EE0010404, by the U.S. National Science Foundation under Award CMMI-1847651, and by Gamma Technologies.

\bibliographystyle{Elsevier_Numbered}
\bibliography{BattBee}
\flushend
\end{document}